\documentclass[final,3p,times]{elsarticle}

\usepackage{amssymb}

\usepackage{lineno}

\usepackage{graphicx}
\usepackage{thm-restate}
\usepackage{latexsym}
\usepackage[utf8]{inputenc}
\usepackage{booktabs}
\usepackage{arydshln}
\usepackage{amssymb,amsthm}
\usepackage{enumitem}
\usepackage{graphicx} 
\usepackage{latexsym,times,amsmath,xspace}
\usepackage{verbatim}
\usepackage{rotating}
\usepackage{environ}
\usepackage{mathtools}
\usepackage{multicol,esvect,relsize}
\usepackage{tikz,hyperref}

\newcommand{\logicFont}[1]{\protect\ensuremath{\mathrm{#1}}\xspace}

\newcommand{\classFont}[1]{\protect\ensuremath{\mathsf{#1}}\xspace}

\newcommand{\D}{\logicFont{D}}

\newcommand{\pind}{\perp\!\!\!\perp}
\newcommand{\cpind}{\perp\!\!\!\perp_{\rm c}}

\newcommand {\pci}[3] {#2~\!\!\pind_{#1}\!\!~#3}
\newcommand {\pmi}[2] {#1~\!\!\pind\!\!~#2}

\newcommand{\SUM}{\mathrm{SUM}}

\newcommand{\x}{\tuple x}

\newcommand{\pcixyz}{\pci{\tuple x}{\tuple y}{\tuple z}}
\newcommand{\pmixy}{\pmi{\tuple x}{\tuple y}}

\newcommand{\add}{\textrm{add}}
\newcommand{\supplement}[3]{\mathrm{S}_{{#1},{#2}}(#3)}
\newcommand{\projection}[2]{\mathrm{Pr}_{#2}(#1)}
\newcommand{\base}{\mathcal{V}}
\newcommand{\idx}{\mathcal{I}}

\newcommand{\PSPACE}{\classFont{PSPACE}}

\newcommand{\NP}{\classFont{NP}}
\newcommand{\NPpoly}{\classFont{NP}/\classFont{poly}}
\newcommand{\bp}[1]{\mathrm{BP}(#1)}

\newcommand{\NPr}{\classFont{NP}_\mathbb{R}}

\newcommand{\PTIME}{\classFont{P}}

\newcommand{\tuple}[1]{\vec{#1}}

\newcommand{\Dom}{\operatorname{Dom}}

\newcommand{\ar}{\operatorname{ar}}

\newcommand{\Fr}{\operatorname{Fr}}
\newcommand{\A}{\mathfrak{A}}

\newcommand{\N}{\mathbb{N}}

\newcommand{\Z}{\mathbb{Z}}
\newcommand{\calS}{\mathcal{S}}

\newcommand{\calL}{\mathcal{L}}
\newcommand{\calC}{\mathcal{C}}
\newcommand{\calD}{\mathcal{D}}

\newcommand{\X}{\mathbb{X}}
\newcommand{\Y}{\mathbb{Y}}
\newcommand{\struc}{\mathrm{Struc}}

\def\dep{=\!\!}
\newcommand{\deps}{\dep(\cdots)}
\newcommand{\sub}{\subseteq}

\newcommand{\Var}{\mathrm{Var}}

\newcommand{\Varfo}{\mathrm{Var_{1}}}
\newcommand{\Varso}{\mathrm{Var_{2}}}

\newcommand{\dfn}{:=}

\newcommand{\enc}{\mathrm{enc}}

\newcommand{\RE}{\mathbb{R}}

\newcommand{\supp}[1]{\mathrm{Supp}(#1)}
\newcommand{\bsupp}[1]{\mathrm{Base}(#1)}
\newcommand{\positive}[1]{\mathrm{Pos}(#1)}
\newcommand{\relweight}[2]{|{#1}_{#2}|_{\mathrm{rel}}}

\newcommand{\existst}{\ddot{\exists}}

\def\dep{=\!\!}
\newcommand{\FO}{{\rm FO}}

\newcommand{\mA}{{\mathfrak A}}

\newcommand{\ESO}{{\rm ESO}}

\newcommand{\eso}[2]{\ESO_{#1}[{#2}]}
\newcommand{\peso}[2]{\mathrm{L}\text{-}\ESO_{#1}[{#2}]}

\newcommand{\esor}[1]{\eso{\RE}{#1}}
\newcommand{\eear}[1]{(\existst^*\exists^*\forall^*)_{\RE}[{#1}]}
\newcommand{\ear}[1]{(\existst^*\forall^*)_{\RE}[{#1}]}
\newcommand{\lear}[1]{\mathrm{L}\text{-}(\existst^*\forall^*)_{\RE}[{#1}]}
\newcommand{\lead}[1]{\mathrm{L}\text{-}(\existst^*\forall^*)_{d[0,1]}[{#1}]}
\newcommand{\leau}[1]{\mathrm{L}\text{-}(\existst^*\forall^*)_{[0,1]}[{#1}]}

\newcommand{\leaeu}[1]{\mathrm{L}\text{-}(\existst^*\forall^*\exists^*)_{[0,1]}[{#1}]}

\newcommand{\pesod}[1]{\peso{d[0,1]}{#1}}
\newcommand{\pesou}[1]{\peso{[0,1]}{#1}}

\newcommand\bigexists{%
  \mathop{\lower0.75ex\hbox{\ensuremath{%
    \mathlarger{\mathlarger{\mathlarger{\mathlarger{\exists}}}}}}}%
  \limits}

\newenvironment{redtext}{\color{red}}{\ignorespacesafterend}
\newenvironment{bluetext}{\color{blue}}{\ignorespacesafterend}

\newcommand{\jonni}[1]{\begin{redtext} Jonni: {#1} \end{redtext}}
\newcommand{\miika}[1]{\begin{bluetext} Miika: {#1} \end{bluetext}}

\usepackage{booktabs}   
\usepackage{subcaption} 
                        
\newtheorem{theorem}{Theorem}
\newtheorem*{subclaim}{Subclaim}
\newtheorem{proposition}[theorem]{Proposition}
\newtheorem{corollary}[theorem]{Corollary}
\newtheorem{claim}{Claim}
\newtheorem{lemma}[theorem]{Lemma}

\newtheorem{definition}[theorem]{Definition}

\newif\iflong 
\longtrue 

\let\oldeq\[
\renewcommand{\[}{\begin{linenomath*}\oldeq}
\let\oldequ\]
\renewcommand{\]}{\oldequ \end{linenomath*}}

\journal{Annals of Pure and Applied Logic}

\begin{document}

\begin{frontmatter}



\title{Tractability frontiers in probabilistic team semantics and existential
               second-order logic over the reals}


\author[miika]{Miika Hannula\fnref{fn1}\corref{cor1}}

     \ead{miika.hannula@helsinki.fi}

\author[jonni1,jonni2]{Jonni Virtema\fnref{fn2}}
     \ead{j.t.virtema@sheffield.ac.uk}

\cortext[cor1]{Corresponding author}
\fntext[fn1]{Supported by the Academy of Finland grant 322795.}
\fntext[fn2]{Supported by the DFG grant VI 1045/1-1.}

\affiliation[miika]{organization={Department of Mathematics and Statistics, University of Helsinki},
            city={Helsinki},
            country={Finland}}

\affiliation[jonni1]{organization={Institut f\"ur Theoretische Informatik, Leibniz Universit\"at Hannover},
            city={Hannover},
            country={Germany}}
            
\affiliation[jonni2]{organization={Department of Computer Science, University of Sheffield},
            city={Sheffield},
            country={United Kingdom}}

\begin{abstract}
Probabilistic team semantics is a framework for logical analysis of probabilistic dependencies. Our focus is on the axiomatizability, complexity, and expressivity of probabilistic inclusion logic and its extensions.
We identify a natural fragment of existential second-order logic with additive real arithmetic that captures exactly the expressivity of probabilistic inclusion logic. We furthermore relate these formalisms to linear programming, and doing so obtain PTIME data complexity for the logics. 
Moreover, 
on finite structures, we show that the full existential second-order logic with additive real arithmetic can only express NP properties.
Lastly, we present a sound and complete axiomatization for probabilistic inclusion logic at the atomic level. 
\end{abstract}




\end{frontmatter}


\section{Introduction}
 
\emph{Metafinite model theory}, introduced by Gr\"adel and Gurevich \cite{GradelG98}, generalizes the approach of \emph{finite model theory} by shifting to two-sorted structures that extend finite structures with another (often infinite) domain with some arithmetic (such as the reals with multiplication and addition), and weight functions bridging the two sorts. 
A simple example of a metafinite structure is a graph involving numerical labels; e.g., a railway network where an edge between two adjacent stations is labeled by the distance between them.
Metafinite structures are, in general, suited for modeling problems that make reference to some numerical domain, be it reals, rationals, or complex numbers.

A particularly important subclass of metafinite structures are the \emph{$\RE$-structures}, which extend finite structures with the real arithmetic on the second sort. The computational properties of $\RE$-structures can be studied with \emph{Blum-Shub-Smale machines} \cite{blum1989} 
(BSS machines for short) 
which are essentially register machines with registers that can store arbitrary real numbers and which can compute rational
 functions over reals in a single time step. 
 
A particularly important related problem is the existential theory of the reals (ETR), which contains all Boolean combinations of equalities and inequalities of polynomials that have real solutions. Instances of ETR are closely related to the question whether a given finite structure can be extended to an $\RE$-structure satisfying certain constraints. Moreover, as we will elaborate more shortly, ETR is also closely related to polynomial time BSS-computations.

\emph{Descriptive complexity theory} for BSS machines and logics on metafinite structures was initiated by Gr\"adel and Meer who showed that $\NPr$ (i.e., non-deterministic polynomial time on BSS machines)  is captured by a variant of existential second-order logic ($\ESO_{\mathbb{R}}$) over $\RE$-structures \cite{GradelM95}. Since the work by Gr\"adel and Meer, others (see, e.g., \cite{CuckerM99,abs-2003-00644,HansenM06,Meer00}) have shed more light upon \emph{the descriptive complexity over the reals} mirroring the development of classical descriptive complexity.

Complexity over the reals can be related to classical complexity by restricting attention to Boolean inputs. The so-called \emph{Boolean part} of $\NPr$, written $\bp{\NPr}$, consists of all those Boolean languages that can be recognized by a BSS machine in non-deterministic polynomial time. In contrast to $\NP$, which is concerned with discrete problems that have discrete solutions, this class captures discrete problems with \emph{numerical} solutions. 
A well studied visibility problem in computational geometry related to deciding existence of numerical solutions is the so-called \emph{art gallery problem}. Here one is asked can a given polygon be guarded by a given number of guards whose positions can be determined with arbitrary precision.
Another typical problem is the recognition of unit distance graphs, that is, to determine whether a given graph can be embedded on the Euclidean plane in such a way that two points are adjacent whenever the distance between them is one.
These problems \cite{AbrahamsenAM18,schaefer13}, and an increasing number of others, have been recognized as complete for the complexity class $\exists \RE$, defined as the closure of ETR with polynomial-time reductions \cite{Schaefer09}. 
The exact complexity of $\exists \RE$ is a major open question; currently it is only known that 
\begin{equation}\label{eq:open}
\NP \leq \exists \RE \leq \PSPACE \, \, \, \text{\cite{Canny88}}.
\end{equation}
Interestingly,  $\exists \RE$ can also be characterized as the Boolean part of $\NPr^0$, written $\bp{\NPr^0}$, where $\NPr^0$ is non-deterministic polynomial time over BSS machines that allow only machine constats $0$ and $1$ \cite{BurgisserC06,SchaeferS17}.  It follows that  $\exists \RE$ captures exactly those properties of finite structures that are definable in $\ESO_{\mathbb{R}}$ (with constants $0$ and $1$). That $\exists \RE$ can be formulated in purely descriptive terms has, to the best of our knowledge, never been made explicit in the literature. Indeed, one of the aims of this paper is to promote a descriptive approach to $\exists \RE$. In particular, our results show that certain additive fragments of $\ESO_{\mathbb{R}}$, which correspond to subclasses of $\exists \RE$, collapse to $\NP$ and $\PTIME$.

 In addition to metafinite structures, the connection between logical definability encompassing numerical structures and computational complexity has received attention  in \emph{constraint databases} \cite{BENEDIKT2003169,GradelK99,Kreutzer00}. A constraint database models (e.g., geometric data) by combining a numerical \emph{context structure} (such as the real arithmetic) with a finite set of quantifier-free formulae defining infinite database relations~\cite{KanellakisKR95}.
 
Renewed interest to logics on frameworks analogous to metafinite structures, and related descriptive complexity theory, is motivated by the need to model inferences utilizing numerical data values in the fields of machine learning and artificial intelligence.  
See e.g. \cite{GroheR17,abs-2009-10574} for declarative frameworks for machine learning utilizing logic, \cite{ConsoleHL20,Torunczyk20} for very recent works on logical query languages with arithmetic, and \cite{JordanK16} for applications of descriptive complexity in machine learning. 
 
In this paper, we focus on the descriptive complexity of logics with so-called \emph{probabilistic team semantics} as well as additive $\ESO_{\mathbb{R}}$.
Team semantics is the semantical framework of modern logics of dependence and independence. Introduced by Hodges \cite{hodges97} and adapted to dependence logic by
V\"a\"an\"anen \cite{vaananen07}, team semantics defines truth in reference to collections of assignments, called \emph{teams}. Team semantics is particularly suitable for a formal analysis of properties, such as the functional dependence between variables, which only arise in the presence of multiple assignments. In the past decade numerous research articles have, via re-adaptations of team semantics, shed more light into the interplay between logic and dependence. A common feature, and limitation, in all these endeavors has been their preoccupation with notions of dependence that are \emph{qualitative} in nature. That is, notions of dependence and independence that make use of quantities, such as conditional independence in statistics, have usually fallen outside the scope of these studies.

The shift to quantitative dependencies in team semantics setting is relatively recent. 
While the ideas of probabilistic teams trace back to the works of Galliani \cite{galliani08} and Hyttinen et al. \cite{Hyttinen15b},
a systematic study on the topic can be traced to \cite{DurandHKMV18,HKMV18}.
In \emph{probabilistic team semantics} the basic semantic units are probability distributions (i.e., \emph{probabilistic teams}).
This shift from set based semantics to distribution based semantics enables probabilistic notions of dependence to be embedded to the framework.
In \cite{HKMV18} probabilistic team semantics was studied in relation to the dependence concept that is most central in statistics: conditional independence. Mirroring \cite{galliani12,GradelM95,kontinenv09} the expressiveness of probabilistic independence logic ($\FO(\cpind)$), obtained by extending first-order logic with conditional independence, was in \cite{HKMV18,abs-2003-00644} characterised in terms of  arithmetic variants of existential second-order logic. In \cite{abs-2003-00644} the data complexity of $\FO(\cpind)$ was also identified in the context of BSS machines and the existential theory of the reals. In \cite{HHKKV19} the focus was shifted to the expressivity hierarchies between probabilistic logics defined in terms of different quantitative dependencies.
Recently, the relationship between the settings of probabilistic and relational team semantics has raised interest in the context of quantum information theory \cite{abramsky2021team,DBLP:journals/corr/abs-2102-10931}.

Another vantage point to quantitative dependence comes from the notion of \emph{multiteam semantics}, defined in terms of multisets of variable assignments called \emph{multiteams}. 
A multiteam can be viewed as a database relation that not only allows duplicate rows (cf. SQL data tables), but also keeps track of the number of times each row is repeated. Multiteam semantics and probabilistic team semantics are close parallels, and they often exhibit similar behavior with respect to their key logics (cf. \cite{DurandHKMV18,wilke20,wilkelode20}). There are also differences, namely because the two frameworks are designed to model different situations. For instance, a probability of a random variable can be halved, but it makes no sense to consider a data row that is repeated two and half times in a data table. For this reason, the so-called split disjunction is allowed to cut an assignment weight into two halves in one framework but not (always) in the other.

 Of all the dependence concepts thus far investigated in team semantics, that of \emph{inclusion} has arguably turned out to be the most intriguing and fruitful. One reason is that \emph{inclusion logic}, which arises from this concept, can only define properties of teams that are decidable in polynomial time \cite{gallhella13}. In contrast, other natural team-based logics, such as dependence and independence logic, capture non-deterministic polynomial time \cite{galliani12,kontinenv09,vaananen07}, and many variants, such as team logic, have an even higher complexity \cite{kontinennu09}. Thus it should come as no surprise if quantitative variants of many team-based logics turn out more complex; in principle, adding  arithmetical operations and/or counting cannot be a mitigating factor when it comes to complexity.

In this paper, we study \emph{probabilistic inclusion logic}, which is the extension of first-order logic with so-called \emph{marginal identity atoms} $x \approx y$ which state that $x$ and $y$ are identically distributed.
Our particular focus is on the complexity and expressivity of \emph{sentences}. It is important, at this point, to note the distinction between formulae and sentences in team-based logics: Formulae describe properties of \emph{teams} (i.e., relations), while sentences describe properties of \emph{structures}. This distinction is even more pointed in probabilistic team semantics, where formulae describe properties \emph{probabilistic teams} (i.e., real-valued probability distributions). 
On the other hand, sentences of logics with probabilistic team semantics can express
variants of important problems that are conjectured not to be expressible in the relational analogues of the logics. Decision problems related to ETR (i.e., the likes of the art gallery problem) are, in particular, these kind of problems.
Another motivation to focus on sentences is our desire to \emph{make comparison} between relational and quantitative team logics. 
As discussed above, the move from relational to quantitative dependence should not in principle make the associated logics weaker. There is, however, no direct mechanism to examine this hypothesis at the formula level, because the team properties of relational and quantitative team logics are essentially incommensurable.
Fortunately this becomes possible at  the sentence level. The reason is that sentences describe only properties of (finite) structures in \emph{both} logical approaches.

The main takeaway of this paper is that there is no drastic difference between a relational team logic and its quantitative variant, as long as the latter makes only reference to \emph{additive} arithmetic. While inclusion logic translates to fixed point logic,  its quantitative variant, probabilistic inclusion logic, seems to require linear programming. Yet, the complexity upper bounds ($\NP$/$\PTIME$) of first-order logic extended with dependence and/or inclusion atoms are preserved upon moving to quantitative variants. In contrast, earlier results indicate that this is not necessarily the case with respect to dependencies whose quantitative expression involves multiplication (such as conditional independence \cite{abs-2003-00644}).

{\bf Our contribution.}
We use strong results from linear programming to obtain the following complexity results over finite structures.
We identify a natural fragment of additive $\ESO_{\RE}$ (that is, \emph{almost conjunctive $\ear{\leq,+,\SUM, 0,1}$}) which captures $\PTIME$ on ordered structures (see page \pageref{almostconjunctive} for a definition). The full additive $\ESO_{\RE}$ is in turn  shown to capture $\NP$. Additionally, we establish that the so-called \emph{loose fragments}, almost conjunctive $\leau{=,\SUM,0,1}$ and $\pesou{=,+,0,1}$, of the aforementioned logics have the same expressivity as probabilistic inclusion logic and its extension with dependence atoms, respectively. The characterizations of $\PTIME$ and $\NP$ hold also for these fragments.
Over open formulae, probabilistic inclusion logic extended with dependence atoms is shown to be strictly weaker than probabilistic independence logic. Moreover, we expand from a recent analogous result by Gr\"adel and Wilke on multiteam semantics \cite{wilke20} and show that probabilistic independence cannot be expressed in any logic that has access to only atoms that are relational or closed under so-called \emph{scaled unions}. In contrast, independence logic and inclusion logic with dependence atoms are equally expressive in team semantics \cite{galliani12}.
We also show that inclusion logic can be conservatively embedded into its probabilistic variant, when restricted to probabilistic teams that are uniformly distributed. From this we obtain an alternative proof through linear systems  (that is entirely different from the original proof of Galliani and Hella \cite{gallhella13}) for the fact that inclusion logic can express only polynomial time properties. Finally, we present a sound and complete axiomatization for marginal identity atoms. This is achieved by appending the axiom system of inclusion dependencies with a symmetricity rule. 

This paper is an extended version of \cite{HanVirJelia21}. Here we include all the proofs that were previously omitted. In addition, the results in Sections \ref{sect:sep} and \ref{sect:ax} are new.

\section{Existential second-order logics on $\RE$-structures}\label{sec:preli}
In addition to finite relational structures, we consider their numerical extensions by adding real numbers ($\RE$) as a second domain sort and functions that map tuples over the finite domain to $\RE$. Throughout the paper structures are assumed to have at least two elements. In the sequel, $\tau$ and $\sigma$ will always denote a finite relational and a finite functional vocabulary, respectively.    
The arities of function variables $f$ and relation variables $R$ are denoted by $\ar(f)$ and $\ar(R)$, resp.
If $f$ is a function with domain $\Dom(f)$ and $A$ a set, we define $f\upharpoonright A$ to be the function with domain   $\Dom(f)\cap A$ that agrees with $f$ for each element in its domain.
Given a finite set $S$, a function $f\colon S\to[0,1]$ that maps elements of $S$ to elements of the closed interval $[0,1]$ of real numbers such that $\sum_{s\in S}f(s)=1$ is called a \emph{(probability) distribution}, and the \emph{support} of $f$ is defined as $\supp{f} \dfn \{s\in S\mid f(s) >  0\}$. Also, $f$ is called \emph{uniform} if $f(s)=f(s')$ for all $s,s'\in\supp{f}$.


 

\begin{definition}[$\RE$-structures]
A tuple
\(
\A = (A,\RE, (R^\A)_{R\in\tau}, 
(g^\A)_{g\in\sigma}),
\)
where the reduct of $\A$ to $\tau$ is a finite relational structure, and 
 each  $g^\A$ is a function from $A^{\ar(g)}$ to $\RE$, is called an \emph{$\RE$-structure of vocabulary $\tau\cup\sigma$}. Additionally, $\A$ is also called (i) an $S$-structure, for $S\sub \RE$, if each $g^\A$ is a function from $A^{\ar(g)}$ to $S$, and (ii) a $d[0,1]$-structure if each $g^\A$ is a distribution.
 We call $\A$ a \emph{finite structure}, if $\sigma=\emptyset$.
\end{definition}

Our focus is on a variant of functional existential second-order logic with numerical terms ($\ESO_\RE$) that is designed to describe properties of $\RE$-structures. As first-order terms we have only first-order variables. For a set   $\sigma$ of function symbols, the set of numerical $ \sigma$-terms $i$ is generated by the following grammar:
 \[
i ::= c\mid f(\vec{x}) \mid i + i \mid i\times i\mid
\SUM_{\vec{y}} \, i,
\]
where $\vec y$ can be any tuple of variables and include variables that do not occur in $i$.
The interpretations of $+,\times,\SUM$ are the standard addition, multiplication, and summation of real numbers, respectively, and $c\in\RE$ is a real constant denoting itself.
In particular, the interpretation $[\SUM_{\vec y} \, i]^\A_s$ of the term $\SUM_{\vec y} \, i$ is defined as follows:
\[
[\SUM_{\vec y}
\, i]^\A_s :=
\sum_{\vec a \in A^{|\vec y|}} [i]^\A_{s[\vec a/\vec y]},
\]
where $[i]^\A_{s[\vec a/\vec y]}$ is an interpretation of the term $i$.
We write $i(\vec{y})$ to mean that the free variables of the term $i$ are exactly the variables in $\vec{y}$. The free variables of a term are defined as usual. In particular, the variables in $\vec{x}$ are not free in $\SUM_{\vec{x}} i(\vec{y})$.
%

\begin{definition}[Syntax of $\ESO_\RE$]\label{def:eso}
Let $O\sub \{+,\times, \SUM\}$,  $E \sub \{=,<,\leq\}$, and $C\sub \RE$.  The set of $\tau\cup\sigma$-formulae of $\esor{O,E,C}$ is defined via the grammar:
\[
\phi ::= \   x=y \,|\, \neg x= y \,|\, i \mathrel e j \,|\, \neg
{i\mathrel e j} \,|\, R(\vec{x}) \,|\, \neg R(\vec{x}) \,|\,
 \phi\land\phi \,|\,
  \phi\lor\phi \,|\,  \exists x\phi \,|\, \forall x \phi \,|\,  \exists f \psi,
\]
where $i$ and $j$ are numerical $\sigma$-terms constructed using operations from $O$ and constants from $C$; $e\in E$; $R\in\tau$ is a relation symbol; $f
 $ is a function variable; $x$, $y$, and $\vec{x}$ are (tuples of) first-order variables; and $\psi$ is a $\tau\cup(\sigma\cup\{f\})$-formula of $\esor{O,E,C}$.  
\end{definition}
The semantics of $\esor{O,E,C}$ is defined via $\RE$-structures
and assignments analogous to first-order logic, however the interpretations of function variables $f$ range over functions $A^{\ar(f)} \to \RE$.
Furthermore, given $S\sub \RE$, we define $\eso{S}{O,E,C}$ as the variant of $\esor{O,E,C}$ in which quantification of functions range over $h\colon A^{\ar(f)} \to S$.
 
\vspace{-2.5mm} 
\paragraph{\textbf{Loose fragment}}
For $S \sub \RE$,
define $\peso{S}{O,E,C}$ as the \emph{loose
fragment} of $\eso{S}{O,E,C}$ in which negated numerical atoms
$\neg {i\mathrel ej}$ are disallowed. 
%

\vspace{-2.5mm} 
\paragraph{\textbf{Almost conjunctive}}\label{almostconjunctive} A formula $\phi\in \ESO_S[O,E,C]$ is \emph{almost conjunctive}, if for every subformula $(\psi_1 \lor \psi_2)$ of $\phi$, no numerical term occurs in $\psi_i$, for some $i\in\{1,2\}$.

\vspace{-2.5mm} 
\paragraph{\textbf{Prefix classes}} For a regular expression $L$ over the alphabet $\{\existst, \exists, \forall \}$, we denote by $L_S[O,E,C]$ the formulae of $\ESO_S[O,E,C]$ in prefix form whose quantifier prefix is in the language defined by $L$, where $\existst$ denotes existential function quantification, and  $\exists$ and $\forall$ first-order quantification.

\vspace{-2.5mm} 
\paragraph{\textbf{Expressivity comparisons}}
Let $\calL$ and $\calL'$ be some logics defined above, and let $X\subseteq \RE$.
For $\phi\in\calL$, define $\struc_{X}(\phi)$ to be the class of pairs $(\mA,s)$ where $\mA$ is an $X$-structure 
and $s$ an assignment
such that $\A\models_s \phi$. 
Define $\struc_{\rm fin}(\phi)$ ($\struc_{\rm ord}(\phi)$, resp.) analogously in terms of finite (finite ordered, resp.) structures. 
%
Additionally, $\struc_{d[0,1]}(\phi)$ is the class of $(\A,s) \in \struc_{[0,1]}(\phi)$ such that each $f^\A$ is a distribution.
%
%
If $X$ is a set of reals or from \{``$d[0,1]$'',``$\mathrm{fin}$'', ``$\mathrm{ord}$''\}, we write $\calL \leq_X \calL'$ if for all formulae $\phi\in \calL$ there is a formula $\psi\in \calL'$ such that $\struc_X(\phi)= \struc_X(\psi)$.
For formulae without free first-order variables, we omit $s$ from the pairs $(\mA,s)$ above.
 %
As usual, the shorthand $\equiv_X$
 stands for $\leq_X$ in both directions.
  For $X=\RE$, we write simply $\leq$ and $\equiv$.

  \vspace{-1mm}
\section{Data complexity of additive $\ESO_\RE$}

On finite structures $\esor{\leq,+, \times,0,1}$ is known to capture the complexity class $\exists\RE$ \cite{BurgisserC06,GradelM95,SchaeferS17}, which lies somewhere between $\NP$ and $\PSPACE$. Here we focus on the additive fragment of the logic. It turns out that the data complexity of the additive fragment is $\NP$ and thus no harder than that of $\ESO$. Furthermore, we obtain a tractable fragment of the logic, which captures $\PTIME$ on finite ordered structures.

\vspace{-1mm}
\subsection{A tractable fragment}\label{sect:almost}

Next we show $\PTIME$ data complexity for almost conjunctive $\eear{\leq,+, \SUM,0,1}$.


\begin{proposition}\label{prop:lp}
Let $\phi$ be an almost conjunctive $\esor{\leq,+, \SUM,0,1}$-formula in which no existential first-order quantifier is in a scope of a universal first-order quantifier. There is a polynomial-time reduction from $\RE$-structures $\mA$ and assignments $s$ to families of systems of linear inequations $\mathcal{S}$ such that $\mA \models_s \phi$ if and only if there is a system $S\in \mathcal{S}$ that has a solution. If $\phi$ has no free function variables, the systems of linear inequations in $\mathcal{S}$ have integer coefficients.
\end{proposition}
\begin{proof}
Fix $\phi$. 
We assume, w.l.o.g.,
that variables quantified in $\phi$ are quantified exactly once, the sets of free and bound variables of $\phi$ are disjoint, and that the domain of $s$ is the set of free variables of $\phi$.
Moreover, we assume that $\phi$ is of the form $\exists \vec{y} \exists \vec{f} \forall  \vec{x} \theta$, where $\vec{f}$ is a tuple of function variables and $\theta$ is quantifier-free. We use $X$ and $Y$ to denote the sets of variables in $\vec{x}$ and $\vec{y}$, respectively, and $\vec{g}$ to denote the free function variables of $\phi$.

We describe a polynomial-time process of constructing a family of systems of linear inequations $\mathcal{S}_{\mA,s}$ from a given $\tau\cup\sigma$-structure $\mA$ and an assignment $s$. We introduce
\begin{itemize}
\item a fresh variable $z_{\vec{a},f}$, for each $k$-ary function symbol $f$ in $\vec{f}$ and $k$-tuple $\vec{a}\in A^k$.
\end{itemize}
In the sequel, the variables $z_{\vec{a},f}$ 
will range over real numbers.

Let $\mA$ be a $\tau\cup\sigma$-structure and $s$ an assignment for the free variables in $\phi$.
In the sequel, each interpretation for the variables in $\vec{y}$ yields a system of linear equations.
Given an interpretation $v\colon Y \rightarrow A$, we will denote by $S_v$ the related system of linear equations to be defined below.
We then set $\mathcal{S}_{\mA,s} \dfn \{ S_v \mid v:  Y \rightarrow A \}$.
%
The system of linear equations $S_v$ is defined as
 \(
S_v \dfn \bigcup_{u\colon X \rightarrow A} S_v^u, 
\)
where $S_v^u$ is defined as follows. Let $s_v^u$ denote the extension of $s$ that agrees with $u$ and $v$. We let $\theta_{v}^u$ denote the formula obtained from $\theta$ by the following simultaneous substitution:
If $(\psi_1\lor \psi_2)$ is a subformula of $\theta$ such that no function variable occurs in $\psi_i$, then $(\psi_1\lor \psi_2)$ is substituted with $\top$, if 
\begin{linenomath*}
\begin{equation}\label{eq:lp2}
\mA\models_{s_v^u} \psi_i,
\end{equation}
\end{linenomath*}
and with $\psi_{3-i}$ otherwise. The set $S_v^u$ is now generated from $\theta_{v}^u$ together with $u$ and $v$. Note that $\theta_{v}^u$ is a conjunction of first-order or numerical atoms $\theta_i$, $i\in I$, for some index set $I$. For each conjunct $\theta_i$ in which some $f\in\vec{f}$ occurs, add  $(\theta_i)_{s_v^u}$ to $S_v^u$, where  $(\psi)_{s_v^u}$ is defined recursively as follows:
\begin{linenomath*}
\begin{align*}
&(\neg \psi)_{s_v^u} \dfn \neg (\psi)_{s_v^u}, &
&(i e j)_{s_v^u} \dfn (i)_{s_v^u}\, e\,  (j)_{s_v^u}, \text{ for each $e\in\{=,<,\leq,+\}$},\\
& (f(\vec{z}))_{s_v^u} \dfn z_{s_v^u(\vec{z}), f}, &
&(\SUM_{\vec{z}} i)_{s_v^u}  \dfn  \sum_{a\in A^{\lvert \vec{z}\rvert}}  (i)_{s_v^u(\vec{a}/\vec{z})},\\
& (g(\vec{z}))_{s_v^u} \dfn  g^\mA(s_v^u(\vec{z})), &
& (x)_{s_v^u} \dfn s_v^u(x), \text{ for every variable $x$}.
\end{align*}
\end{linenomath*}
Let $\theta^*$ be the conjunction of those conjuncts of $\theta_{v}^u$ in which no $f\in\vec{f}$ occurs.
If $\mA\not\models_{s_v^u} \theta^*$, remove $S_v$ from $\mathcal{S}_{\mA,s}$.

Since $\phi$ is fixed, it is clear that $\mathcal{S}_{\mA,s}$ can be constructed in polynomial time with respect to $\lvert \mA \rvert$.
 Moreover, it is straightforward to show that there exists a solution for some $S\in \mathcal{S}_{\mA,s}$ exactly when $\mA\models_s \phi$.

Assume first that there exists an $S\in \mathcal{S}_{\mA,s}$ that has a solution.  Let $w\colon Z \rightarrow \RE$, where $Z \dfn \{z_{\vec{a},f} \mid f\in \vec{f} \text{ and }  \vec{a} \in A^{\ar(f)}\}$, be the function given by a solution for $S$. By construction, $S=S_v$, for some $v\colon Y\rightarrow A$. Let $\mA'$ be the expansion of $\mA$ that interprets each $f\in\vec{f}$ as the function $\vec{a}\mapsto w(z_{\vec{a},f})$. By construction,
\(
\mA'\models_{s^u_v} \theta^u_v
\)
for every $u\colon X\rightarrow A$. Now, from \eqref{eq:lp2} and the related substitutions, we obtain that
\(
\mA'\models_{s^u_v} \theta 
\)
for every $u\colon X\rightarrow A$, and hence
$\mA'\models_{s_v} \forall x_1 \ldots \forall x_n  \theta$.
From this $\mA \models_s \phi$ follows.

For the converse, assume that $\mA \models_s \phi$. Hence there exists an extension $s_v$ of $s$ and an expansion $\mA'$ of $\mA$ such that $\mA'\models_{s_v} \forall x_1 \ldots \forall x_n  \theta$. 
Now, by construction, it follows that $S_v \in \mathcal{S}_{\mA,s}$ and
\(
\mA'\models_{s^u_v} \theta^u_v,
\)
for every $u\colon X\rightarrow A$. Moreover, it follows that the function defined by $z_{\vec{a},f} \mapsto f^{\mA'}(\vec{a})$, for $f\in\vec{f}$ and $\vec{a}\in A^{\ar(f)}$, is a solution for $S_v$.
   
\end{proof}
The above proposition could be strengthened by relaxing the almost conjunctive requirement in any way such that \eqref{eq:lp2} can be still decided (i.e., it suffices that the satisfaction of $\psi_i$s do not depend on the interpretations of the functions in $\vec{f}$).
\begin{theorem}\label{thm:ptimedataesor}
The data complexity of almost conjuctive $\esor{\leq,+, \SUM, 0,1}$-formulae without free function variables and where no existential first-order quantifiers are in a scope of a universal first-order quantifier is in $\PTIME$.
\end{theorem}
\begin{proof}
Fix an almost conjuctive $\esor{\leq,+, \SUM, 0,1}$-formula $\phi$ of relational vocabulary $\tau$ of the required form. Given a $\tau\cup\emptyset$ structure $\mA$ and an assignment $s$ for the free variables of $\phi$, let $\mathcal{S}$ be the related polynomial size family of polynomial size systems of linear inequations with integer coefficients given by Proposition \ref{prop:lp}. Deciding whether a system of linear inequalities with integer coefficients has solutions can be done in polynomial time \cite{zbMATH03644821}. Thus checking whether there exists a system of linear inequalities $S\in\mathcal{S}$ that has a solution can be done in $\PTIME$ as well, from which the claim follows. 
   
\end{proof}

We will later show that probabilistic inclusion logic captures $\PTIME$   on finite ordered structures (Corollary \ref{cor:pincP}) and can be translated to almost conjunctive $\leau{\leq,\SUM, 0,1}$ (Lemma \ref{lem:fromFOprob}). Hence already almost conjuctive $\lear{\leq,\SUM, 0,1}$ captures $\PTIME$.

\begin{corollary}\label{cor:learptime}
Almost conjunctive $\lear{\leq,\SUM, 0,1}$ captures $\PTIME$ on finite ordered structures. 
\end{corollary}

\subsection{Full additive $\ESO_\RE$}
The goal of this subsection is to prove the following theorem:

\begin{theorem}\label{thm:addesonp}
$\esor{\leq,+,\SUM,0,1}$ captures $\NP$ on finite structures.
\end{theorem}

First observe that $\SUM$ is definable in $\esor{\leq,+,0,1}$:  
Already $\esor{=}$ subsumes $\ESO$, and thus we may assume a built-in successor function $S$ and its associated minimal and maximal elements $\min$ and $\max$ on $k$-tuples over the finite part of the $\RE$-structure. Then, for a $k$-ary tuple of variables $\tuple x$, $\SUM_{\tuple x} i$ agrees with $f(\max)$, for any function variable $f$ satisfying $f(\min)= i(\tuple x \mapsto \min)$ and $f(S(\tuple x)) = f(\tuple x) + i(S(\tuple x))$. 
%

As $\esor{\leq,+, 0,1}$ subsumes $\ESO$, by Fagin's theorem, it can express all $\NP$ properties.
Thus we only need to prove that any $\esor{\leq,+, 0,1}$-definable property of finite structures is recognizable in $\NP$.  The proof relies on (descriptive) complexity theory over the reals. The fundamental result in this area is that existential second-order logic over the reals ($\esor{\leq,+,\times,(r)_{r\in \RE}}$) corresponds to non-deterministic polynomial time over the reals ($\NP_\RE$) for BSS machines \cite[Theorem 4.2]{GradelM95}. 
To continue from this, some additional terminology is needed. We refer the reader to \ref{sect:app} (or to the textbook \cite{BSSbook}) for more details about BSS machines.
Let  $C_{\RE}$ be a complexity class over the reals.
\begin{itemize}[nosep]
 \item $C_{\text{add}}$ is $C_\RE$ restricted to \emph{additive} BSS machines (i.e., without multiplication).
 \item $C^0_\RE$ is $C_\RE$ restricted to BSS machines with machine constants $0$ and $1$ only.
 \item $\bp{C_\RE}$ is $C_\RE$ restricted to languages of strings that contain only $0$ and $1$. 
 \end{itemize}
 A straightforward adaptation of \cite[Theorem 4.2]{GradelM95} yields the following theorem. 
\begin{theorem}[\cite{GradelM95}]\label{addnp}
$\esor{\leq,+,0,1}$ captures $\NP^0_{\emph{add}}$ on $\RE$-structures.
\end{theorem}
If we can establish that $\bp{\NP^0_{\text{add}}}$, the so-called \emph{Boolean part} of $\NP^0_{\text{add}}$, collapses to $\NP$, we have completed the proof of Theorem \ref{thm:addesonp}.
Observe that another variant of this theorem readily holds; 
 $\esor{=,+,(r)_{r\in \RE}}$-definable properties of $\RE$-structures are recognizable in $\NP_{\add}$ branching on equality, which in turn, over Boolean inputs, collapses to $\NP$ 
\cite[Theorem 3]{Koiran94}. 
 Here, restricting branching to equality is crucial. With no restrictions in place (the BSS machine by default branches on inequality and can use arbitrary reals as machine constants) $\NP_{\add}$ equals $\NPpoly$ over Boolean inputs \cite[Theorem 11]{Koiran94}. Adapting arguments from \cite{Koiran94}, we show next that disallowing machine constants other than $0$ and $1$, but allowing branching on inequality, is a mixture that leads to a collapse to $\NP$. 

\begin{restatable}{theorem}{readdboole}\label{addboole}
$\bp{\NP^{0}_{\emph{add}}} = \NP$.
\end{restatable}
\begin{proof}
Clearly  $\NP\leq \bp{\NP^{0}_{\add}}$; a Boolean guess for an input $\vec{x}$ can be constructed by comparing to zero each component of a real guess $\vec{y}$, and a polynomial-time Turing computation can be simulated by a polynomial-time BSS computation.

For the converse, let $L\sub \{0,1\}^*$ be a Boolean language that belongs to $\bp{\NP^{0}_\add}$; we need to show that $L$ belongs also to $\NP$. Let $M$ be a BSS machine such that its running time is bounded by some polynomial $p$, and for all Boolean inputs $\vec x \in \{0,1\}^*$, $\vec{x}\in L$ if and only if there is $\vec{y}\in \RE^{p(|x|)}$ such that $M$ accepts $(\vec{x},\vec{y})$.



We describe a non-deterministic algorithm that decides $L$ and runs in polynomial time. Given a Boolean input $\vec{x}$ of length $n$,  first guess the outcome of each comparison in the BSS computation; this guess is a Boolean string $\vec{z}$ of length $p(n)$. 
Note that each configuration of a polynomial time BSS computation can be encoded by a real string of polynomial length.
During the BSS computation the value of each coordinate of its configuration is a linear function on the constants $0$ and $1$, the input $\vec{x}$, and the real guess $\vec{y}$ of length $p(n)$. Thus it is possible to construct in polynomial time a system $\mathcal{S}$ of linear inequations on $\vec{y}$ of the form
\begin{linenomath*}
\begin{equation}\label{eq:linearsystem}
\sum_{j=1}^{p(n)} a_{ij}y_j  \leq 0 \quad  (1\leq i \leq m)  \quad \text{and} \quad \sum_{j=1}^{p(n)} b_{ij}y_j < 0 \quad (1\leq i \leq l),
\end{equation}
\end{linenomath*}
where $a_{ij}\in \mathbb{Z}$, such that $\vec{y}$ is a (real-valued) solution to $\mathcal{S}$ if and only if $M$ accepts $(\vec{x},\vec{y})$ with respect to the outcomes $\vec{z}$. In \eqref{eq:linearsystem}, the variables $y_j$ stand for elements of the real guess $\vec{y}$, and $m+l$ is the total number of comparisons.  Each comparison generates either a strict or a non-strict inequality, depending on the outcome encoded by $\vec{z}$.

Without loss of generality we may assume additional constraints of the form $y_j \geq 0$ $(1 \leq j \leq p(n))$ (cf. \cite[p. 86]{oldtextbook}). 
Transform then $\mathcal{S}$ to another system of inequalities $\mathcal{S'}$ obtained from $\mathcal{S}$ by replacing strict inequalities in \eqref{eq:linearsystem} by
\begin{linenomath*}
\begin{equation*}
\sum_{j=1}^{p(n)} b_{ij}y_j + \epsilon \leq 0 \quad (1\leq i \leq l)\quad \text{and} \quad \epsilon \leq 1,
\end{equation*}
\end{linenomath*}
Then determine the solution of the linear program: maximize $(\tuple 0,1)(\tuple y,\epsilon)^T$ subject to $\mathcal{S'}$ and $(\tuple y,\epsilon)\geq 0$. If there is no solution or the solution is zero, then reject; otherwise accept. Since $\mathcal{S'}$ is of polynomial size and linear programming is in polynomial time \cite{zbMATH03644821}, the algorithm runs in polynomial time. Clearly, the algorithm accepts $\vec{x}$ for some guess $\vec{z}$ if and only if $\vec{x}\in L$. 
\end{proof}

\section{Probabilistic team semantics and additive $\ESO_\RE$}
\subsection{Probabilistic team semantics}

Let $D$ be a finite set of first-order variables and $A$ a finite set. A \emph{team} $X$ is a set of assignments from $D$ to $A$. A \emph{probabilistic team} is a distribution $\X\colon X\rightarrow [0,1]$, where $X$ is a finite team.
Also the empty function is considered a probabilistic team.
We call $D$ the variable domain of both $X$ and $\X$, written $\Dom(\X)$ and $\Dom(X)$. $A$ is called the \emph{value domain} of $X$ and $\X$.

Let $\X:X \to [0,1]$ be a probabilistic team,  $x$ a variable, $V\sub \Dom(\X)$ a set of variables, and $A$ a set. The \emph{projection} of $\X$ on $V$ is defined as $\projection{\X}{V} :X\upharpoonright V \to [0,1]$ such that $s \mapsto \sum_{t\upharpoonright V=s} \X(t)$, where $X\upharpoonright V :=\{t\upharpoonright V \mid t\in X\}$.  Define $\supplement{x}{A}{\X}$ as the set of all probabilistic teams $\Y$ with variable domain $\Dom(\X)\cup \{x\}$ such that $\projection{\Y}{\Dom(\X)\setminus\{x\}}=\projection{\X}{\Dom(\X)\setminus\{x\}}$ and $A$ is a value domain of $\Y\upharpoonright \{x\}$. We denote by $\X[A/x]$ the unique $\Y\in \supplement{x}{A}{\X}$ such that
 \[\Y(s) = \frac{\projection{\X}{\Dom(\X)\setminus \{x\}}(s\upharpoonright \Dom(\X)\setminus \{x\})}{|A|}.\]
 If $x$ is a fresh variable, then this equation becomes
  \(\Y(s(a/x)) = \frac{\X(s)}{|A|}.\)
 We also define
$
X[A/x] \dfn \{s(a/x) \mid s\in X, a\in A \}
$, and write $\X[a/x]$ and $X[a/x]$ instead of $\X[\{a\}/x]$ and $X[\{a\}/x]$, for singletons $\{a\}$.

Let us also define some function arithmetic. 
Let $\alpha$ be a real number, and $f$ and $g$ be functions from a shared domain into real numbers. The scalar multiplication $\alpha f$ is a function defined by $(\alpha f)(x):=\alpha f(x)$. The addition  $f+g$ is defined as $(f+g)(x)=f(x)+g(x)$, and the multiplication $fg$ is defined as $(fg)(x):=f(x)g(x)$. In particular, if $f$ and $g$ are probabilistic teams and $\alpha +\beta =1$, then $\alpha f + \beta g$ is a probabilistic team.
 
 We define first probabilistic team semantics for first-order formulae. As is customary in the team semantics context, we restrict attention to formulae in negation normal form.
{If $\phi$ is a first-order formula,  we write $\phi^\bot$ for the equivalent formula obtained from $\neg \phi$ by pushing the negation in front of atomic formulae. If furthermore $\psi$ is some (not necessarily first-order) formula, we then use a shorthand $\phi\to \psi$ for the formula $\phi^{\bot} \lor (\phi \land \psi)$.}

\begin{definition}[Probabilistic team semantics]\label{def:semantics}
	Let $\A$ be a $\tau$-structure over a finite domain $A$, and $\X\colon X \to [0,1]$ a proba\-bi\-listic team. The satisfaction relation $\models_\X$ for first-order logic is defined as follows:
	\begin{tabbing}
		\hspace{-3mm}left \= $\A \models_{\X} (\psi \land \theta)$ \= $\Leftrightarrow$ \= $\forall s\in X: s(\tuple x) \in R^{\A}$\kill
		\> $\A \models_{\X} l$ \> $\Leftrightarrow$ \> $\forall s\in \supp{\X}: \A\models_s l$, where $l$ is a literal\\ 
		\> $\A\models_{\X} (\psi \land \theta)$ \> $\Leftrightarrow$ \> $\A \models_{\X} \psi \text{ and } \A \models_{\X} \theta$\\
		\> $\A\models_{\X} (\psi \lor \theta)$ \> $\Leftrightarrow$ \> $\A\models_{\Y} \psi \text{ and } \A \models_\Z \theta,
		\text{ for some probabilistic teams $\Y$ and $\Z$, and}$\\
		\>\>\>$\alpha\in[0,1]\text{ such that } \alpha\Y + (1-\alpha )\Z =\X$ \\
		\> $\A\models_{\X} \forall x\psi$ \> $\Leftrightarrow$ \> $\A\models_{\X[A/x]} \psi$\\
		\> $\A\models_{\X} \exists x\psi$ \> $\Leftrightarrow$ \> $\A\models_{\Y} \psi \text{ for some }\Y\in\supplement{x}{A}{\X}$
	\end{tabbing}
\end{definition}
The satisfaction relation $\models_s$ denotes the Tarski semantics of first-order logic. If $\phi$ is a \emph{sentence} (i.e., without free variables), then $\A$ \emph{satisfies} $\phi$, written $\A \models \phi$, if $\A\models_{\X_\emptyset} \phi$, where $\X_\emptyset$ is the distribution that maps the empty assignment to $1$.


We make use of a generalization of probabilistic team semantics where the requirement of being a distribution is dropped. A \emph{weighted team} is any non-negative weight function $\X\colon X \rightarrow \RE_{\geq 0}$. Given a first-order formula $\alpha$, we write $\X_{\alpha}$ for the restriction of the weighted team $\X$ to the assignments of $X$ satisfying $\alpha$ (with respect to the underlying structure). Moreover, the \emph{total weight} of a weighted team $\X$ is $|\X|:=\sum_{s\in X} \X(s)$. 

\begin{definition}[Weighted semantics]\label{def:weightedsemantics}
Let $\A$ be a $\tau$-structure over a finite domain $A$, and $\X\colon X \to \RE_{\geq 0}$ a weighted team. The satisfaction relation $\models^w_\X$ for first-order logic is defined exactly as in Definition \ref{def:semantics}, except that for $\lor$ we define instead:
\[
\mA \models^w_\X  (\psi \lor \theta) \quad\Leftrightarrow\quad  \A\models_{\Y} \psi \text{ and } \A \models_\Z \theta \text{ for some }\Y,\Z \text{ s.t. } \Y + \Z =\X.
\]
\end{definition}
%
We consider logics with the following atomic dependencies:
%
\begin{definition}[Dependencies]
Let $\A$ be a finite structure with universe $A$, $\X$ a weighted team, and $X$ a team.
\begin{itemize}
\setlength\itemsep{0em}
\item \textbf{Marginal identity and inclusion atoms}. If $\tuple x,\tuple y$ are variable sequences of length $k$, then $\tuple x \approx \tuple y$ is a \emph{marginal identity atom} and $\tuple x \sub \tuple y$ is an \emph{inclusion atom}  with satisfactions defined as:
\begin{linenomath*}
\begin{align*}
\A \models^w_{\X} \vec{x} \approx \vec{y} &\Leftrightarrow \lvert {\X}_{\vec{x}=\vec{a}} \rvert =  \lvert {\X}_{\vec{y}=\vec{a}} \rvert \text{ for each $\vec{a}\in A^k$},\\
\A\models_X \tuple x \sub \tuple y &\Leftrightarrow \text{for all }s\in X \text{ there is }s'\in X\text{ such that }s(\tuple x)=s'(\tuple y).
\end{align*}
\end{linenomath*}
\iflong 
\item \textbf{Probabilistic independence atom.} If $\tuple x,\tuple y,\tuple z$ are variable sequences, 
then $\pci{\tuple x}{\tuple y}{\tuple z}$ is a \emph{probabilistic (conditional) independence atom} with satisfaction defined as:
\begin{linenomath*}
\begin{align*}
	\A\models_{\X} \pcixyz \label{def1}
\end{align*}
\end{linenomath*}
if for all $s\colon\Var(\tuple x\tuple y\tuple z) \to A$ it holds that
\[
\lvert {\X}_{\tuple x\tuple y= s(\tuple x\tuple y)} \rvert \cdot  \lvert {\X}_{\tuple x\tuple z=s(\tuple x\tuple z)} \rvert = \lvert {\X}_{\tuple x\tuple y\tuple z=s(\tuple x\tuple y\tuple z)} \rvert \cdot \lvert {\X}_{\tuple x=s(\tuple x)} \rvert .
\]
We also write $\pmi{\tuple x}{\tuple y}$ for the \emph{probabilistic marginal independence atom}, defined as $\pci{\emptyset}{\tuple x}{\tuple y}$.

\fi

\item \textbf{Dependence atom}. For a sequence of variables $\tuple x$ and a variable $y$, $\dep(\tuple x, y)$ is a \emph{dependence atom} with satisfaction defined as: 
\[\A\models_X \dep(\tuple x, y) \Leftrightarrow \text{for all }s,s'\in X: \text{ if }s(\tuple x)=s'(\tuple x)\text{, then }s(y)=s'(y).\]
\end{itemize}
For probabilistic teams $\X$, the satisfaction relation is written without the superscript $w$. 
\end{definition}
Observe that any dependency $\alpha$ over team semantics can also be interpreted in probabilistic team semantics: $\A\models_\X \alpha$ iff $\A\models_{\supp{\X}} \alpha$. 
For a list $\calC$ of dependencies, we write $\FO(\calC)$ for the extension of first-order logic with the dependencies in $\calC$. The logics $\FO(\approx)$ and $\FO(\sub)$, in particular, are called \emph{probabilistic inclusion logic} and \emph{inclusion logic}, respectively. Furthermore, \emph{probabilistic independence logic} is denoted by $\FO(\cpind)$, and its restriction to probabilistic marginal independence atoms by $\FO(\pind)$. 
We write $\Fr(\phi)$ for the set free variables of $\phi \in \FO(\calC)$, defined as usual.  We conclude this section with a list of useful equivalences. We omit the proofs, which are straightforward structural inductions (\ref{i:alternateprop} was also proven in \cite{HHKKV19} and \ref{i:flatness} follows from \ref{i:conservative} and the flatness property of team semantics).

\begin{proposition}\label{propositions}
Let $\phi\in \FO(\calC)$, $\psi\in\FO(\approx,\calC)$, and $\theta \in \FO$,
where $\calC$ is a list of dependencies over team semantics. Let $\A$ be a structure, $\X$ a weighted team, and $r$ any positive real. The following equivalences hold: 
\vspace{-3mm}
\begin{multicols}{2}
\begin{enumerate}[label=(\roman*)]
	\item $\A\models^w_\X \phi \;\Leftrightarrow\; \A\models_{\supp{\X}} \phi$. 
	\label{i:conservative}
	\item  $\mA\models^w_\X \psi \;\Leftrightarrow\; \mA \models_{\frac{1}{\lvert \X \rvert} \X}~\psi$. \label{i:alternateprop}
	\item\label{i:scaling} $\A\models^w_\X \psi \;\Leftrightarrow\;  \A\models^w_{r\X} \psi$. \label{i:scaling}
	\item $\A\models^w_\X \psi \;\Leftrightarrow\; \A\models^w_{\X \upharpoonright V} \psi$, where
	 $\Fr(\psi)\sub V$. \label{i:locality}
	\item $\A\models^w_\X \theta \; \Leftrightarrow \; \A\models_s \theta$, for all
	 $s\in \supp{X}$. \label{i:flatness}
\end{enumerate}
\end{multicols}
\end{proposition}

\subsection{Expressivity of probabilistic inclusion logic}

We turn to the expressivity of probabilistic inclusion logic and its extension with dependence atoms. In particular, we relate these logics to existential second-order logic over the reals. We show that probabilistic inclusion logic extended with dependence atoms captures a fragment  in which arithmetic is restricted to summing. Furthermore, we show that leaving out dependence atoms is tantamount to restricting to sentences in almost conjunctive form with $\existst^*\forall^*$ quantifier prefix.

\vspace{-3mm}
\paragraph{Expressivity comparisons.}
 Fix a list  of atoms $\calC$ over probabilistic team semantics. For a probabilistic team $\X$ with variable domain $\{x_1, \ldots ,x_n\}$ and value domain $A$, the function $f_\X:A^n\to [0,1]$ is defined as the probability distribution such that $f_\X(s(\tuple x))=\X(s)$ for all  $s\in X$. For a formula $\phi \in \FO(\calC)$
of vocabulary $\tau$ and with free variables $\{x_1, \ldots ,x_n\}$, the class $\struc_{d[0,1]}(\phi)$ is defined as the class of $d[0,1]$-structures $\A$ over $\tau\cup\{f\}$ such that
 $(\A\upharpoonright \tau)\models_{\X} \phi$, where  $f_\X=f^\A$ and $\A\upharpoonright \tau$ is the finite $\tau$-structure underlying $\A$.
Let $\calL$ and $\calL'$ be two logics of which one is defined over (probabilistic) team semantics. We write $\calL\leq \calL'$ 
 if for every formula 
  $\phi\in \calL$ there is $\phi'\in \calL'$ such that $\struc_{d[0,1]}(\phi)= \struc_{d[0,1]}(\phi')$; again, $\equiv$ is a shorthand for $\leq $ both ways.

\begin{theorem}\label{thm:FOprob}
The following equivalences hold:
\begin{enumerate}[nosep,label=(\roman*)]
\item $\FO(\approx,\deps)\equiv \pesou{=,+,0,1}$. \label{i:equivyks}
\item $\FO(\approx)\equiv \text{ almost conjunctive }\leau{=,\SUM,0,1}$. \label{i:equivkaks}
\end{enumerate}
\end{theorem}


We divide the proof of Theorem \ref{thm:FOprob} into two parts. In Section \ref{sect:prob-to-eso} we consider the direction from probabilistic team semantics to existential second-order logic over the reals, and in Section \ref{sect:eso-to-prob} we shift attention to the converse direction.
%
%
In order to simplify the presentation in the forthcoming subsections, we start by 
showing how to replace existential function quantification by  distribution quantification. 
 The following lemma in its original form includes multiplication (see \cite[Lemma 6.4]{abs-2003-00644}) but works also without it. %

\begin{lemma}[\cite{abs-2003-00644}]\label{lem:lics}
$\pesou{=,+,0,1}\equiv_{d[0,1]}  \pesod{=,\SUM}$. 
\end{lemma} 
The proof, however, does not preserve the almost conjunctive form. That case is dealt with separately in Proposition \ref{prop:remove-d}. 
As shown next, we can utilize in this proposition the fact that the real constants $0$ and $1$ are definable in almost conjunctive $\lead{=,\SUM}$.

\begin{lemma}\label{lem:realzero}
$ \pesod{=,\SUM} \equiv_{\RE} \pesod{=,\SUM,0,1}$. The same holds when both logics are restricted to almost conjunctive formulae of the prefix class $\existst^*\forall^*$.
\end{lemma}
\begin{proof}
Any formula $\theta$ involving $0$ or $1$ can be equivalently expressed as follows:
\[
\exists n \exists f \exists h \forall x \forall y \forall z \Big(f(x)=h(x,x) \land \big(y=z \lor \theta(h(y,z)/0, n/1)\big)\Big),
\]
where $n$ is nullary.
    \end{proof}

%

\begin{restatable}{proposition}{reremoved}\label{prop:remove-d}
	$\pesou{=,\SUM,0,1}\equiv_{[0,1]}  \pesod{=,\SUM}$. The same holds when both logics are restricted to almost conjunctive formulae of the prefix class $\existst^*\forall^*$.
\end{restatable}

\begin{proof}
		The $\geq$-direction is trivial. We show the $\leq$-direction, which is similar to the proof of \cite[Lemma 6.4]{abs-2003-00644}. By Lemma \ref{lem:realzero}
		we may assume that almost conjunctive $\lead{=,\SUM}$ (as well as $\pesod{=,\SUM}$) contains real constants $0$ and $1$. 		
Suppose $\phi$ is some formula in $\pesou{=,\SUM,0,1}$.
Let $k$ be the maximal arity of any function variable/symbol appearing in $\phi$. 
The total sum of the weights of any interpretation of a function occurring in $\phi$ on a given structure, whose finite domain is of size $n$, is at most $n^k$.
We now show how to obtain from $\phi$ an equivalent formula in $\pesod{\SUM,=,0,1}$; the idea is to scale all function weights by $1/n^k$. Note first that the value $1/n^k$ can be expressed via a $k$-ary distribution variable $g$ as follows:
\[
\exists g \forall \vec{x} \vec{y} \, g(\vec{x}) = g(\vec{y}) 
\]
Below, we write $\frac{1}{n^k}$ instead of $g(\vec{x})$.

Suppose $\phi$ is of the form $\exists f_1\dots f_m \forall\vec{x} \theta$, where $\theta$ is quantifier free, and let $g_1,\dots, g_t$ be the list of (non-quantified) function symbols of $\phi$. Define
\[
\phi' \dfn \exists f'_1\dots f'_m g'_1 \dots g'_t \forall\vec{x}\vec{x}' \,( \psi \land \theta'),
\]
where each $f'_j$ ($g'(j)$, resp.) is an $\ar(f_j)+1$-ary ($\ar(g_j)+k+1$-ary, resp.) distribution variable and $\psi$ and $\theta'$ are as defined below. The universally quantified variables $\vec{x}'$ list all of the newly introduced variables of the construction below. The formula $\psi$ is used to express that each $f'_j$ ($g'(j)$, resp.) is an $1/n^k$-scaled copy of $f_j$ ($g(j)$, resp.). That  is, $\psi$ is defined as the formula 
\[
\bigwedge_{i \leq m} f'_j(\vec{y},y_l) \leq \frac{1}{n^k} \land \bigwedge_{i \leq t} \big(g'_j(\vec{y},\vec{z}, z_l)=g'_j(\vec{y}, \vec{z'}, z'_l) \land \SUM_{\vec{z}} g'_j(\vec{y}, \vec{z}, z_l) = g_j(\vec{y})\big),
\]
where $y_l$ and $z_l$ (here and below) denote the last elements of the tuples $\vec{y}$ and $\vec{z}$, respectively.\footnote{For a $0$-ary function $f$, a construction  $f'(\vec{z},z_l)= f'(\vec{z}',z'_l)$ can be used instead.}
Finally $\theta'$ is obtained from $\theta$ by replacing expressions of the form $f_j(\vec{y})$ and $g_j(\vec{y})$ by $f'_j(\vec{y},y_l)$ and $g_j(\vec{y},\vec{z}, z_l)$, resp., and the real constant $1$ by $\frac{1}{n^k}$.
A straightforward inductive argument on the structure of formulae yields that, over $[0,1]$-structures, $\phi$ and $\phi'$ are equivalent. Note that $\phi'$ is an almost conjunctive formula of the prefix class $\existst^*\forall^*$, if $\phi$ is.
    \end{proof}


\subsection{From probabilistic team semantics to existential second-order logic}\label{sect:prob-to-eso}

Let $c$ and $d$ be two distinct constants. Let $\phi(\tuple x)\in \FO(\approx,\deps)$ be a formula whose free variables are from the sequence $\tuple x=(x_1, \ldots ,x_n)$. We now construct recursively an $ \pesou{=,\SUM,0,1}$-formula $\phi^*(f)$ that contains one free $n$-ary function variable $f$.
In this formula, a probabilistic team $\X$ is represented as a function $f_\X$ such that $\X(s) = f_\X(s(x_1), \ldots , s(x_n))$.
\begin{enumerate}[label=(\arabic*)]
\item If $\phi(\tuple x)$ is a first-order literal, then
\[
\phi^*(f):= \forall \vec{x} \big( f(\vec{x})=0 \lor \phi(\tuple x) \big).
\]
\item If $\phi(\tuple x)$ is a dependence atom of the form $\dep(\tuple x_0,x_1)$, then
\[
\phi^*(f):= \forall \vec{x} \, \tuple x'\big( f(\vec{x})=0 \lor f(\vec{x'})=0\lor \tuple x_0 \neq \tuple x'_0\lor x_1=x'_1 \big).
\]
\item
If $\phi(\tuple x)$ $\tuple x_0 \approx \tuple x_1$, where $\tuple x=\tuple x_0\tuple x_1\tuple x_2$, then
\[
\phi^*(f):=\forall \tuple y \, \SUM_{\vec{x}_1,\tuple x_2} f(\tuple y,\tuple x_1,\tuple x_2) =  \SUM_{\vec{x}_0,\tuple x_2} f(\tuple x_0,\tuple y,\tuple x_2).
\]
\item If $\phi(\tuple x)$ is of the form $\psi_0(\tuple x)\land\psi_1(\tuple x)$, then 
\[
\phi^*(f) := \psi_0^*(f)\land\psi_1^*(f).
\]
\item If $\phi(\tuple x)$ is of the form $\psi_0(\tuple x)\lor\psi_1(\tuple x)$, then 
\[
\phi^*(f) \dfn \exists g\forall \tuple x \, (\SUM_yg(\tuple x,y)=f(\tuple x) \land \forall y(y=c\lor y=d \lor g(\tuple x,y)=0)\wedge \psi^*_0(g^c) \land \psi^*_1(g^d)),
\]
where $g^i$ is of the same arity as $f$ and defined as $g^i(\tuple x) \dfn g(\tuple x,i)$.

\item If $\phi(\tuple x)$ is $\exists y \psi(\tuple x,y)$, then
\[ 
\phi^*(f) \dfn \exists g \big( ( \forall \vec{x} \, \SUM_y g(\vec{x},y) = f(\vec{x})) \land \psi^*(g) \big).
\]
\item If $\phi(\tuple x)$ is of the form $\forall y \psi(\tuple x,y)$, then 
\[
\phi^*(f) \dfn \exists g \big(  \forall \vec{x} (\forall y \forall z  g(\vec{x},y) = g(\vec{x},z) \land \, \SUM_y g(\vec{x},y) = f(\vec{x})     ) \land \psi^*(g) \big).
\]
\end{enumerate}
This translation leads to the following lemma, 
%
%
 
\begin{restatable}{lemma}{refromFOprob}\label{lem:fromFOprob}
The following hold:
\begin{enumerate}[nosep, label=(\roman*)]
\item $\FO(\approx,\deps)\leq \leau{= ,\SUM,0,1}$.\label{claim:1} 
\item $\FO(\approx,\deps)\leq \text{ almost conjunctive } \leaeu{= ,\SUM,0,1}$. \label{claim:2}
\item $\FO(\approx)\leq \text{ almost conjunctive } \leau{= ,\SUM,0,1}$. \label{claim:3}
\end{enumerate}
\end{restatable}

\begin{proof}
By item \ref{i:alternateprop} of Proposition \ref{propositions}, we may use weighted semantics (Definition \ref{def:weightedsemantics}).
Then, a straightforward induction shows that for all structures $\A$ and non-empty weighted teams $\X\colon X\to [0,1]$, with variable domain $\tuple x$, such that $|\X|\leq 1$,
\begin{linenomath*}
\begin{equation}\label{eq:esotrans}
\A\models^w_{\X} \phi(\tuple x) \iff (\A,f_\X)\models \phi^*(f).
\end{equation}
\end{linenomath*}
Furthermore, the extra constants $c$ and $d$ can be discarded. Define $\psi(f)$ as
\begin{linenomath*}
\begin{equation}\label{eq:psif}
\exists f'\forall cd \forall \tuple x  \big(f'(\tuple x,c,d) = f(\tuple x) \land \big(c\neq d \to \phi^{**}(f') )\big),
\end{equation}
\end{linenomath*}
where $\phi^{**}(f')$ is obtained from $\phi^*(f)$ by replacing function terms $f(t_1, \ldots ,t_n)$ with $f'(t_1, \ldots ,t_n, c, d)$. 
There are only existential function and universal first-order quantifiers in \eqref{eq:psif}.
By pushing these quantifiers in front, and by swapping the ordering of existential and universal quantifiers (by increasing the arity of function variables and associated function terms), we obtain a sentence $\psi^*(f)\in \lead{= ,\SUM,0,1}$ which, if substituted for $\phi^*(f)$, satisfies \eqref{eq:esotrans}.

Let us then  turn to the items of the lemma.
\begin{enumerate}[nosep, label=(\roman*)]
\item  The claim readily holds.
\item The claim follows if the translation for dependence atoms $\dep(\tuple x_0,x_1)$ and $\tuple x=\tuple x_0 x_1\tuple x_2$ is replaced by
\[\phi^*(f):= \forall \tuple x_0 \exists x_1  \SUM_{\tuple x_2}f(\vec{x})= \SUM_{x_1\tuple{x}_2} f(\vec{x}).\]
We conclude that $\phi^*(f)$ interprets the dependence atom in the correct way and it preserves the almost conjunctive form and the required prefix form.
\item  For the claim, it suffices to drop the translation of the dependence atom. 
\end{enumerate}
   
\end{proof}


This completes the ``$\leq$'' direction of Theorem \ref{thm:FOprob}. For \ref{i:equivyks}, this follows from \ref{claim:1} of Lemma \ref{lem:fromFOprob}, Proposition \ref{prop:remove-d}, and Lemma \ref{lem:lics}. For \ref{i:equivkaks}, only \ref{claim:3} of Lemma \ref{lem:fromFOprob} is needed.

 Recall from Proposition \ref{prop:lp} that almost conjunctive $\eear{\leq ,+,\SUM,0,1}$ is in PTIME in terms of data complexity. Since dependence logic captures $\NP$ \cite{vaananen07}, the previous lemma indicates that we have found, in some regard, a maximal tractable fragment of additive existential second-order logic. That is, dropping either the requirement of being almost conjunctive, or that of having the prefix form $\existst^*\exists^*\forall^*$, leads to a fragment that captures $\NP$; that $\NP$ is also an upper bound for these fragments follows by Theorem \ref{thm:addesonp}.

\begin{corollary}\label{main}
$\FO(\approx, \deps)$ captures $\NP$ on finite structures.
\end{corollary}

\subsection{From existential second-order logic to probabilistic team semantics}\label{sect:eso-to-prob}

Due to Lemma \ref{lem:lics} and Proposition \ref{prop:remove-d}, our aim is to translate $\pesod{=,\SUM}$ and almost conjunctive $\pesod{=,\SUM}$ to $\FO(\approx, \deps)$ and $\FO(\approx)$, respectively.
The following lemmas imply that we may restrict attention to formulae in Skolem normal form.\footnote{Lemma \ref{lemmakaks} was first presented in \cite[Lemma 3]{HKMV18} in a form that included multiplication. 
 We would like to thank Richard Wilke for noting that the construction used in \cite{HKMV18} to prove this lemma 
 had an element that yields circularity.
 Furthermore, we would like to than Joni Puljuj\"arvi for noting another issue which is circumvented by Lemma \ref{lem:safeness}.
	
}

We first need to get rid of all numerical terms whose interpretation does not belong to the unit interval. The only source of such terms are summation terms of the form $\SUM_{\vec{x}} i(\vec{y})$, where $\vec{x}$ is a sequence of variables that contain a variable $z$ not belonging to $\vec{y}$; we call such instances of $z$ \emph{dummy-sum instances}.
 For example, the summation term $\SUM_{x} n$, where $n$ is the nullary distribution and $x$ a dummy-sum instance, is always interpreted as the cardinality of the model's domain.
\begin{lemma}\label{lem:safeness}
For every $\pesod{=,\SUM}$-formula $\phi$ there exists an equivalent formula without dummy-sum instances.
\end{lemma}
\begin{proof}
Let $k$ be the number of dummy sum-instances in $\phi$. Without loss of generality,  we may assume that each dummy sum-instance is manifested using a distinct variable in $\vec{v}=(v_1,\dots,v_k)$, whose only instance in $\phi$ is the related dummy sum-instance. It is straighforward to check that for any structure $\mA$ with cardinality $n$, the interpretation $t^\mA$ of any term $t$ appearing in $\phi$ is at most $n^k$. 
 
 We start the translation $\psi \mapsto \psi^*$ by scaling each function $f$ occurring in $\phi$ by $\frac{1}{n^k}$ 
 as follows.
 Define $f(\vec{x})\mapsto f^*(\vec{x},\vec{v})$.
 For Boolean connectives, $=$, $\SUM$, and first-order 
 {quantification} the translation is homomorphic. In the case for existential function quantification, the functions are scaled by increasing their arity by $k$ and stipulating that their weights are distributed evenly over the arity extension:
 \[
 \exists f \psi \mapsto \exists f^* \big(\forall \vec{x}\, \vec{v}\,\vec{w}\, f^*(\vec{x},\vec{v}) = f^*(\vec{x},\vec{w}) \land \phi^* \big).
 \]
 Let $f_1,\dots,f_t$
 be the list of free
 function variables of $\phi$ with arities $\lvert \vec{x}_1\rvert, \dots, \lvert \vec{x}_t\rvert$, respectively.
Now, define
\[
\phi^+ \dfn  \exists f^*_1 \dots f^*_t \Big( \bigwedge_{l\leq t} \big( \forall \vec{x}_l\,\SUM_{\vec{v}}f_l^*(\vec{x}_l,\vec{v}) = f_l(\vec{x}_l) \land \forall \vec{x}_l \, \vec{v} \, \vec{w} \, f^*_l(\vec{x}_l,\vec{v}) = f^*_l(\vec{x}_l,\vec{w}) \big)  \land \exists \vec{v} \, \phi^*\Big).
\]
It is now straightforward to check that $\phi^+$ and $\phi$ are equivalent, and that there are no dummy-sum instances in $\phi^+$.
\end{proof}

\begin{restatable}
{lemma}{relemmakaks}\label{lemmakaks}
For every formula $\phi\in\pesod{=,\SUM}$ there is a formula $\phi^*\in\lead{=,\SUM} $ such that 
$\struc_{d[0,1]}{(\phi)}=\struc_{d[0,1]}{(\phi^*)}$, 
and any second sort identity atom in $\phi^*$ is of the form 
$f_i(\tuple w) =\SUM_{\tuple v}f_j(\tuple u,\tuple v)$ for
distinct $f_i$ and $f_j$ of which at least one is quantified. Furthermore, $\phi^*$ is almost conjunctive if $\phi$ is almost conjunctive and in $\lead{=,\SUM}$.
\end{restatable}

\begin{proof}
By the previous lemma, we may assume without loss of generality that $\phi$ does not contain any dummy-sum instances. That is, any summation term occurring in $\phi$ is of the form $\SUM_{\tuple v}i(\tuple u\tuple v)$, where it is to be noted that the variables of $\vec{v}$ occur free in the term $i$. This, in particular, implies that the terms of $\phi$ can be captured by using distributions.

First we define for each second sort term $i(\tuple x)$ a special formula $\theta_i$ defined recursively using fresh function symbols $f_i$ as follows:
\begin{itemize}
\item If $i(\tuple u)$ is $ g(\tuple u)$ where $g$ is a function symbol, then $\theta_i$ is defined as $ f_i(\tuple u)= g(\tuple u)$. (We may intepret $g(\tuple u)$ as $\SUM_{\emptyset} g(\tuple u)$).
\item If $i(\tuple u)$ is $\SUM_{\tuple v}j(\tuple u\tuple v)$, then $\theta_i$ is defined as $ \theta_j\wedge f_i(\tuple u)=\SUM_{\tuple v}f_j(\tuple u\tuple v)$.
\end{itemize}
The translation $\phi\mapsto \phi^*$  then proceeds recursively on the structure of $\phi$. By Lemma \ref{lem:realzero} we may use the real constant $0$ in the translation.
\begin{enumerate}[label=(\roman*)]
\item If $\phi$ is $  i(\tuple u)=j(\tuple v)$, then $\phi^*$ is defined as $\exists \tuple f (f_i(\tuple u)=f_j(\tuple v)  \wedge \theta_i\wedge 
\theta_j)$ 
where $\tuple f$ lists the function symbols $f_k$ for each subterm $k$ of $i$ or $j$.
\item If $\phi$ is an atom or negated atom of the first sort, then $\phi^*:=\phi$.
\item If $\phi$ is $ \psi_0\circ\psi_1$ where $\circ\in \{\vee,\wedge\}$, $\psi^*_0$ is $\exists\tuple  f_0\forall \tuple x_0\theta_0$, and $\psi^*_1$ is $\exists\tuple  f_1\forall\tuple x_1\theta_1$, then $\phi^*$ is defined as $\exists \tuple f_0\tuple f_1\forall \tuple x_0\tuple x_1 (\theta_0\circ \theta_1)$.
\item If $\phi$ is $\exists y\psi$ where $\psi^*$ is $\exists\tuple  f\forall \tuple x\theta$, then $\phi^*$ is defined as $\exists g\exists\tuple  f\forall \tuple x\forall y(g(y)=0\vee \theta)$. \label{eks}
\item  Suppose $\phi$ is $\forall y\psi$ where $\psi^*$ is $\exists\tuple  f\forall \tuple x\theta$. Let $\tuple g$ list the free distribution variables in $\phi$. Then $\phi^*$ is defined as 
\begin{align*}
\exists\tuple  f^*\exists \tuple g^*
\forall yy'\forall \tuple x \Big( &\bigwedge_{g^*\in \tuple g^*}\big( g^*(y,\tuple x)=g^*(y',\tuple x) \land \SUM_{y} g^*(y,\tuple x)=g(\tuple x)\big) \wedge \\
&\bigwedge_{f^*\in \tuple f^*} \big(f^*(y,\tuple x)=f^*(y',\tuple x) \big) \land \theta^* \Big),
\end{align*}
where $\tuple f^*$ ($\tuple g^*$, resp.) is obtained from $\tuple f$  ($\tuple g$, resp.) by replacing each $f$  ($g$, resp.) from $\tuple f$  ($\tuple g$, resp.) with $f^*$ ($g^*$, resp.) such that $\ar(f^*)=\ar(f)+1$ ($\ar(g^*)=\ar(g)+1$, resp.), 
and $\theta^*$ is obtained from $\theta$ by replacing all function terms $f(\tuple z)$ ($g(\tuple z)$, resp.) with $f^*(y, \tuple z)$ ($g^*(y, \tuple z)$, resp.).

\item If $\phi$ is $\exists f\psi$ where $\psi^*$ is $\exists\tuple  f\forall \tuple x\theta$, then $\phi^*$ is defined as $\exists f\psi^*$.
\end{enumerate}
It is straightforward to check that $\phi^*$ is of the correct form and equivalent to $\phi$. What happens in (v) is that instead of guessing for all $y$ some distribution $f_y$ with arity $\ar(f)$, we guess a single distribution $f^*$ with arity $\ar(f)+1$ such that $f^*(y,\tuple u)=\frac{1}{|A|}\cdot f_y(\tuple u)$, where $A$ is the underlying domain of the structure. 
Similarly, we guess a distribution $g^*$ for each free distribution variable $g$ such that $g^*(y,\tuple u)=\frac{1}{|A|}\cdot g(\tuple u)$. Observe that  case \ref{eks} does not occur if $\phi$ is in $\lead{\SUM, =}$; in such a case, a straightforward structural induction  shows that $\phi^*$ is almost conjuctive if $\phi$ is.
    \end{proof}

Using the obtained normal form for existential second-order logic over the reals we now proceed to the translation. This translation is similar to one found in \cite{HKMV18}, with the exception that probabilistic independence atoms cannot be used here.


\begin{restatable}{lemma}{refromESOf}\label{lemma:fromESOf}
Let $\phi(f) \in\lead{=,\SUM} $ be of the form described in Lemma \ref{lemmakaks}, with one free variable $f$.
Then there is a formula $\Phi(\tuple x)\in \FO(\approx,\deps)$ such that for all structures $\A$ and probabilistic teams $\X := f^\A$,
\(
\A\models_{\X} \Phi \iff (\A,f)\models \phi.
\)
Furthermore, if $\phi(f)$ is almost conjunctive,  then $\Phi(\tuple x)\in \FO(\approx)$.
\end{restatable}

\begin{proof}
By item \ref{i:alternateprop} of Proposition \ref{propositions}, we can use weighted semantics in this proof. Without loss of generality each structure is enriched with two distinct constants $c$ and $d$; such constants are definable in $\FO(\approx,\deps)$ by $\exists cd (\dep(c)\land \dep(d) \land c \neq d)$, and for almost conjunctive formulae they are not needed.

Let $ \phi(f)=\exists \tuple f \forall \tuple x \,\theta(f,\tuple x)\in\lead{=,\SUM}$ be of the form described in the previous lemma, with one free variable $f$.
In what follows, we build $\Theta$ inductively from $\theta$, and then
 let
 \begin{linenomath*}
\begin{equation*}
\Phi \dfn  \exists \tuple y_1 \ldots \exists\tuple y_n \forall \tuple x \,\Theta(\tuple x, \tuple y_1, \ldots ,\tuple y_n),  
\end{equation*}
\end{linenomath*}
  where 
  $\tuple y_i$ are sequences of variables of length $\ar(f_i)$. 
 Let $m\dfn \lvert \tuple x \rvert$. We show the following claim: For $M\sub A^m$ and weighted teams $\Y=\X'[M/\tuple x]$, where the domain of $\X'$ extends that of $\X$ by $\tuple y_1, \ldots, \tuple y_n$, 
  \begin{linenomath*}
  \begin{equation}\A\models^w_{\Y} \Theta
  \text{ iff } (\A, f,f_1, \ldots ,f_n) \models \theta(\tuple a)\text{ for all }\tuple a\in M,
  \end{equation}
  \end{linenomath*} 
  where $f_i:= \X' \upharpoonright \tuple y_i$.
   Observe that the claim implies that $\A\models^w_{\X} \Phi$ iff $\A\models \phi(f)$.

  Next, we show the claim by structural induction on the construction of $\Theta$:

 \begin{enumerate}[label=(\arabic*),nosep]
 \item  If $\theta$  is a literal of the first sort, we let $\Theta: =\theta$, and the claim readily holds.
  \item If $\theta$ is of the form
 $f_i(\tuple x_i) = \SUM_{\tuple x_{j0}}f_j(\tuple x_{j0}\tuple x_{j1})$, let $\Theta:= \exists \alpha\beta\psi $ for $\psi$ given as 
 \begin{linenomath*}
 \begin{equation}\label{eq:theta}
 (\alpha=x\leftrightarrow \tuple x_i=\tuple y_i) \wedge (\beta=x \leftrightarrow  \tuple x_{j1}=\tuple y_{j1} )\wedge \tuple x \alpha\approx \tuple x \beta,
 \end{equation}
 \end{linenomath*}
where $x$ is any variable from $\tuple x$, and the first-order variable sequence $\tuple y_j$ that corresponds to function variable $f_j$ is thought of as a concatenation of two sequences $\tuple y_{j0}$ and $\tuple y_{j1}$ whose respective lenghts are $|\tuple x_{j0}|$ and $|\tuple x_{j1}|$.

 Assume first that $\text{for all }\tuple a\in M$, we have $(\A, f,f_1, \ldots ,f_n) \models \theta(\tuple a)$, that is, $f_i(\tuple a_i) = \SUM_{\tuple x_{j0}}f_j(\tuple x_{j0}\tuple a_{j1})$. To show that $\Y$ satisfies $\Theta$, let $\Z$ be an extension of $\Y$ to variables $\alpha$ and $\beta$ such that it satisfies the first two conjuncts of \eqref{eq:theta}. Observe that $\Z$ satisfies $\tuple x \alpha\approx \tuple x \beta$ if for all $\tuple a\in M$, $\Z_{\tuple x = \tuple a}$ satisfies 
 $ \alpha\approx  \beta$.  For a probabilistic team $\X$ and a first-order formula $\alpha$, we write $\relweight{\X}{\alpha}$ for the relative weight $|\X_\alpha| / |\X|$.
 
 Now, the following chain of equalities hold:
   \begin{linenomath*}
   \begin{align*}
&\relweight{\Z}{ \tuple x \alpha= \tuple ax}=\relweight{ \Y}{\tuple x\tuple x_i=\tuple a \tuple y_i}=\relweight{\Y}{\tuple x \tuple y_i=\tuple a \tuple a_i}=\relweight{\Y}{\tuple x=\tuple a}\cdot \relweight{\Y}{\tuple y_i=\tuple a_i}=\\
&  \relweight{\Y}{\tuple x=\tuple a}\cdot f_i(\tuple a_i)=
 \relweight{\Y}{\tuple x=\tuple a}\cdot \SUM_{\tuple x_{j0}}f_j(\tuple x_{j0}\tuple a_{j1})=\relweight{\Y}{\tuple x=\tuple a}\cdot \relweight{\Y}{\tuple y_{j1}=\tuple a_{j1}}\\
 &\relweight{\Y}{\tuple x\tuple y_{j1}=\tuple a\tuple a_{j1}}=
\relweight{\Y}{\tuple x\tuple x_{j1}=\tuple a\tuple y_{j1}}=\relweight{\Z}{\tuple x\beta=\tuple a x}.
\end{align*}
\end{linenomath*}
Note that the absolute weights $|\Y|$ and $|\Z|$ are equal.
The third equality then follows since $\tuple x $ and $\tuple y_i$ are independent by the construction of $\Y$. It is also here that we need relative  instead of absolute weights.
Thus $\alpha$ and $\beta$ agree with $x$ in $\Z_{\tuple x = \tuple a}$ for the same weight. Moreover, $x$ is some constant $a$ in $\Z_{\tuple x = \tuple a}$, and whenever $\alpha$ or $\beta$ disagrees with $x$, it can be mapped to another constant $b$ that is distinct from $a$. It follows that $\Z_{\tuple x = \tuple a}$ satisfies $\alpha \approx \beta$, and thus we conclude that $\Y$ satisfies $\Theta$.

For the converse direction, assume that $\Y$ satisfies $\Theta$, and let $\Z$ be an extension of $\Y$ to $\alpha$ and $\beta$ satisfying \eqref{eq:theta}.
Then for all $\tuple a\in M$,   $\Z_{\tuple x=\tuple a}$ satisfies $\alpha \approx \beta$ and thereby for all $\tuple a\in M$,
\begin{linenomath*}
\begin{align*}
&\relweight{\Y}{\tuple x =\tuple a}\cdot f_i(\tuple a_i)=
  \relweight{\Z}{\tuple x\alpha=\tuple a x}=  \relweight{\Z}{\tuple x\beta=\tuple a x}=
              \relweight{\Y}{\tuple x=\tuple a }\cdot \SUM_{\tuple x_{k}} f_j(\tuple x_{k},\tuple a_{l}).
 \end{align*} 
 \end{linenomath*}
 For the second equality, recall that $x$ is a constant in $\Z_{\tuple x=\tuple a}$.
 Thus $(\A, f,f_1, \ldots ,f_n) \models \theta(\tuple a)\text{ for all }\tuple a\in M$, which concludes the induction step.

  \item If $\theta$ is $\theta_0\wedge \theta_1$, let $\Theta \dfn \Theta_0\wedge \Theta_1$. The claim follows by the induction hypothesis.
  
 \item If $\theta$ is $\theta_0\vee \theta_1$,  let
 \(
 \Theta \dfn \exists z  \Big(\dep(\tuple x,z) \land \big((\Theta_0\wedge z=c) \vee (\Theta_1 \wedge  z=d)\big)\Big).
 \)
 
 Alternatively, if $\theta_0$ contains no numerical terms, let
\( \Theta \dfn \theta_0 \lor ( \theta_0^{\neg}  \land\Theta_1),\)
where $\theta_0^{\neg} $ is obtained from $\neg \theta_0$ by pushing $\neg$ in front of atomic formulae.

Assume first that $(\A, f,f_1, \ldots ,f_n) \models \theta_0\vee\theta_1$ for all $\tuple a \in M$. Then $M$ can be partitioned to disjoint $M_0$ and $M_1$ such that
 \begin{linenomath*}
 \begin{equation}\label{eq:tai}
 (\A, f,f_1, \ldots ,f_n) \models \theta_i\text{ for all }\tuple a \in M_i.
 \end{equation}
 \end{linenomath*}
  We have two cases:
 \begin{itemize} 
 \item Suppose $\phi(f)$ is not almost conjunctive. Let
  $\Z$ be the extension of $\Y$ to $z$ such that $s(z)=c$ if $s(\tuple x)$ is in $M_0$, and otherwise $s(z)=d$, where $s$ is any assignment in the support of $\Z$. Consequently, $\Z$ satisfies $\dep(\tuple x,z)$. Further, the induction hypothesis implies that $\A\models^w_{\Y_i} \Theta_i$, where $\Y_i:= X'[M_i/\tuple x]$. Since $\frac{|M_0|}{|M|} \Y_0 = \Z_{\tuple z = c}$ and $\frac{|M_1|}{|M|}\Y_1 =  \Z_{\tuple z = d}$, we obtain $\A\models^w_{\Z_{\tuple z = c}} \theta_0$ and $\A\models^w_{\Z_{\tuple z = d}} \Theta_1$ by item \ref{i:scaling} of Proposition \ref{propositions}. We conclude that $\Z$ satisfies  $(\Theta_0\wedge z=0) \vee (\Theta_1 \wedge  z=1)$, and thus $\Y$ satisfies $\Theta$.   
 
 \item Suppose $\phi(f)$ is almost conjunctive. Without loss of generality $\theta_0$ contains no numerical terms. Then $\A \models_{\X'[M_0/\tuple x]} \theta_0$ by flatness (i.e., \ref{i:flatness} of Proposition \ref{propositions}). We may assume that $M_0$ is the maximal subset of $M$ satisfying \eqref{eq:tai}, in which case we also obtain $\A \models_{\X'[M_1/\tuple x]} \theta^\neg_0$ by flatness. Furthermore,  $\A \models_{\X'[M_1/\tuple x]} \Theta_1$ by induction hypothesis.
 
 \end{itemize} 
 The converse direction is shown analogously in both cases. This concludes the proof.

\end{enumerate}
   
\end{proof}


The ``$\geq$'' direction of item \ref{i:equivyks} in Theorem \ref{thm:FOprob} follows by Lemmata \ref{lem:lics}, \ref{lemmakaks}, and \ref{lemma:fromESOf}; that of item \ref{i:equivkaks} follows similarly, except that Proposition \ref{prop:remove-d} is used instead of Lemma \ref{lem:lics}. This concludes the proof of Theorem \ref{thm:FOprob}.

\section{Interpreting inclusion logic in probabilistic team semantics}
Next we turn to the relationship between inclusion and probabilistic inclusion logics. The logics are comparable for, as shown in Propositions \ref{propositions}, team semantics embeds into probabilistic team semantics conservatively. The seminal result by Galliani and Hella shows that inclusion logic captures PTIME over ordered structures  \cite{gallhella13}. We  show that restricting to finite structures, or uniformly distributed probabilistic teams, inclusion logic is in turn subsumed by probabilistic inclusion logic. There are two immediate consequences for this.
 First, the result by Galliani and Hella readily extends to probabilistic inclusion logic. Second, their result obtains an alternative, entirely different proof through linear systems.
  
We utilize another result of Galliani stating that inclusion logic is equiexpressive with \emph{equiextension logic} \cite{galliani12}, defined as the extension of first-order logic with \emph{equiextension} atoms $\tuple x_1 \bowtie \tuple x_2:= \tuple x_1 \sub \tuple x_2 \land \tuple x_2 \sub \tuple x_1$. In the sequel, we relate equiextension atoms to probabilistic inclusion atoms. 


For a natural number $k\in\N$ and an equiextension atom $\tuple x_1 \bowtie \tuple x_2$, where $\tuple x_1$ and $\tuple x_2$ are variable tuples of length $m$, define $\psi^k(\tuple x_1, \tuple x_2)$ as
\begin{linenomath*}
\begin{align}\label{kaannos}
 \forall \tuple u\exists v_1v_2 \forall \tuple{z_0} \exists \tuple z (
&
(\tuple x_1=\tuple u \leftrightarrow v_1=y) \land (\tuple x_2=\tuple u \leftrightarrow v_2=y) \,\land\\
&(\tuple{z_0}=\tuple y \to \tuple z=\tuple y) \land (\neg \tuple z = \tuple y \vee  \tuple u  v_1 \approx \tuple u v_2)),\nonumber
\end{align}
\end{linenomath*}
where $\tuple z $ and $\tuple{z_0}$ are variable tuples of length $k$, and $\tuple y$ is obtained by concatenating $k$ times some variable $y$ in $\tuple u$. 
 Intuitively \eqref{kaannos} expresses that a probabilistic team $\X$, extended with universally quantified $\tuple u$, decomposes to $\Y+\Z$, where $\Y(s) = f_s \X(s)$ for some variable coefficient $f_s\in [\frac{1}{n^k}, 1]$, and $|\Y_{\tuple x_1 = \tuple u}| = |\Y_{\tuple x_2 =\tuple u}|$, for any $\tuple u$. Thus \eqref{kaannos}  implies that $\tuple x_1 \bowtie \tuple x_2$.
  On the other hand, $\tuple x_1 \bowtie \tuple x_2$ implies \eqref{kaannos} if each assignment weight $\X(s)$  equals $g_s|\X|$ for some $g_s\in [\frac{1}{n^k}, 1]$. In this case, one finds the  decomposition $\Y+\Z$ by balancing the weight differences between values of $\tuple x_1$ and $\tuple x_2$. More details are provided in the proof of the next lemma.
\begin{restatable}{lemma}{rehelpinglemma}\label{helpinglemma}
Let $k$ be a positive integer, $\mA$ a finite structure with 
 universe $A$ of size $n$, and $\X:X \to \RE_{\geq 0}$ a weighted team. 
  \begin{enumerate}[label=(\roman*)]
  \item\label{it:pos} Suppose $\mA\models^w_\X \tuple x_1 \bowtie \tuple x_2$, $|\X_{\tuple x_1 =\tuple x_2}|=0$, and  
  $\X(s) \geq \frac{|\X|}{n^k}$ for all $s \in \supp{\X}$.  Then $\mA\models^w_\X \phi^k(\tuple x, \tuple y)$.
  \item If $\mA\models^w_\X \phi^k(\tuple x, \tuple y)$, then $ \mA\models^w_\X \tuple x_1 \bowtie \tuple x_2$.
  \end{enumerate}
\end{restatable}
\begin{proof}
(i) 
Observe that 
$
\X[A/\tuple u]= \frac{1}{n^m} \X^*,
$
 where $\X^*$ is defined as the sum $\X[\tuple a_1/\tuple u] + \ldots + \X[\tuple a_l/\tuple u]$, and $\tuple a_1, \ldots ,\tuple a_l$ lists all elements in $A^m$. By Proposition \ref{propositions}\ref{i:scaling} it suffices to show that $\X^*$ satisfies the formula obtained by removing the outermost universal quantification of $\psi^k$.  
 By Proposition \ref{thm:closure} it suffices to show that each $\X[\tuple a_i/\tuple u]$ individually satisfies the same formula.
 Hence fix a tuple of values $\tuple b\in A^m$ and define $\Y\dfn \X[\tuple b / \tuple u]$. We show that $\Y$ satisfies
\begin{linenomath*}
\begin{align}\label{kaannos2}
 \exists v_1v_2 \forall \tuple{z_0} \exists \tuple{z_1} (&(\tuple x_1=\tuple b \leftrightarrow v_1=c) \land (\tuple x_2=\tuple b \leftrightarrow v_2=c)\,\land\\
&(\tuple{z_0}=\tuple c \to \tuple{z_1}=\tuple c) \land ( \tuple{z_1} = \tuple c \to   v_1 \approx v_2)). \notag
\end{align}
\end{linenomath*}
Observe that we have here fixed $\tuple u \mapsto \tuple b$ and $y\mapsto c$, where $c$ is some value in $\tuple b$.
We have also
 removed $\tuple u$ from the marginal identity atom in \eqref{kaannos}, for it has a fixed value in $\Y$.
  
 Fix some $d \in A$ that is distinct from $c$, and denote by $Y$ be the support of $\Y$. For existential quantification over $v_i$, extend $s\in Y$ by $v_i \mapsto c$ if $s(\tuple x_i)=\tuple b$, and otherwise by $v_i \mapsto d$, so as to satisfy the first two conjuncts. Denote by $\Y':Y'\to \RE_{\geq 0}$ the weighted team, where $Y'$ consists of these extensions, and the weights are inherited from $\Y$. 
 
 

Observe that $\Y'(s) \geq \frac{|\X|}{n^{k}}$ for all $s\in \supp{\Y'}$. Fix $i\in \{1,2\}$, and assume that $|\X_{\tuple x_i=\tuple b}| >0$. Then $|\X_{\tuple x_i=\tuple b}|\geq \frac{|\X|}{n^k}$, and thus using $|\X_{\tuple x_1 =\tuple x_2}|=0$ and $|\X|=|\Y'|$ we obtain
 \[
 w_i \dfn |\Y'_{v_i=c \land v_{3-i} = d}|=|\X_{\tuple x_i=\tuple b \land \tuple x_{3-i} \neq \tuple b}|=
  |\X_{\tuple x_i=\tuple b}|\geq 
 \frac{|\Y'|}{n^k}.
 \]
  Since $\X \models \tuple x_1 \bowtie \tuple x_2$, we obtain that
 $w_1$ and $w_2$ are either both zero or both at least $\frac{|\Y'|}{n^{k}}$.

 Next, let us describe the existential quantification of $\tuple{z_1}$ (later we show how the universal quantification of $\tuple{z_0}$ can be fitted in). The purpose of this step is to balance the possible weight difference between $|\Y'_{\tuple x_1 =\tuple b}|$ and $|\Y'_{\tuple x_2 =\tuple b}|$, which in turn is tantamount to balancing $|\Y'_{\tuple v_1 =c\land v_2 =d}|$ and $|\Y'_{v_1 =d \land v_2 =c}|$.
For $s'\in Y'$,
\begin{enumerate}[label=(\roman*)]
\item if $s'(v_1) = c$ and $s'(v_2) = d$, allocate respectively $\frac{w_{2}}{|\Y'|}$ and $1-\frac{w_{2}}{|\Y'|}$ of the weight of $s'$ to $s'(\tuple c/\tuple{z_1})$ and $s'(\tuple d/\tuple{z_1})$;
\item if $s'(v_1) = d$ and $s'(v_2) = c$, allocate respectively $\frac{w_{1}}{|\Y'|}$ and $1-\frac{w_{1}}{|\Y'|}$ of the weight of $s'$ to $s'(\tuple c/\tuple{z_1})$ and $s'(\tuple d/\tuple{z_1})$; or
\item otherwise, allocate the full weight of $s'$ to $s'(\tuple c/\tuple{z_1})$.
\end{enumerate}
Denote by $\Z$ the probabilistic team obtained this way, and define $\Z'\dfn \Z_{\tuple z_1 = \tuple c}$. We observe that 
\[
|\Z'_{v_1 = c\land v_2 = d} |= |\Z'_{v_1 = d\land v_2 = c} |= \frac{w_1w_{2}}{|\Y'|}.
\]
Furthermore, $|\Z'_{v_1 = c\land v_2 = c} |=0$ and  hence $|\Z'_{v_1 = d\land v_2 = d} |= |\Z'|-\frac{2w_1w_{2}}{|\Y'|}$.
 We conclude that $\Z'$ satisfies $   v_1 \approx v_2$, whence $\Z$ satisfies $\tuple{z_1} = \tuple c \to   v_1 \approx v_2$.
 
 Finally, let us return to the universal quantification of $\tuple{z_0}$, which precedes the existential quantification of $\tuple z$ in \eqref{kaannos2}. The purpose of this step is to enforce that for each $s\in \supp{\Y'}$, the extension $s(\tuple c / \tuple z_1)$ takes a positive weight.
 Observe that $\frac{w_i}{|\Y'|}$ is either zero or at least $\frac{1}{n^{k}}$, for $w_i$ is either zero or at least $\frac{|\Y'|}{n^{k}}$. Furthermore, note that
  universal quantification distributes $\frac{1}{n^{k}}$ of the weight of $s'$ to $s'(\tuple c/\tuple{z_0})$. Thus the weight of $s'$ can be distributed in such a way that both the conditions (i)-(iii) and the formula $\tuple{z_0}=\tuple c \to \tuple{z_1}=\tuple c$ simultaneously hold. This concludes the proof of case \ref{it:pos}.
 
(ii) Suppose that the assignments in $X$ mapping $\tuple x_1$ to $\tuple b$ have a positive total weight in $\X$. By symmetry, it suffices to show that the assignments in $X$ mapping $\tuple x_2 $ to $\tuple b$ also have a positive total weight in $\X$. 
 By assumption there is an extension $\Z$ of $\X[\tuple b / \tuple u]$ satisfying the quantifier-free part of \eqref{kaannos2}. It follows that the total weight of assignments in $\Z$  that map $v_1 $ to $c$ is  positive. Consequently, by $\tuple{z_0}=\tuple c \to \tuple{z_1}=\tuple c$ where $\tuple{z_0}$ is universally quantified, a positive fraction of these assignments maps also $\tuple{z_1}$ to $\tuple c$. This part of $\Z$ is allocated to $v_1 \approx v_2$, and thus the weights of assignments mapping $v_2$ to $c$ is positive as well. But then, going  backwards, we conclude that the total weight of assignments mapping $\tuple x_2$ to $\tuple b$ is positive, which concludes the proof.    
\end{proof}

We next establish that inclusion logic is subsumed by probabilistic inclusion logic at the level of sentences.
\begin{theorem}\label{thm:sentprobinc}
$\FO(\sub)\leq \FO(\approx)$ with respect to sentences. 
\end{theorem}
\begin{proof}
As $\FO(\sub)\equiv \FO(\bowtie)$ (\cite{galliani12}), it suffices to show $\FO(\bowtie)\leq \FO(\approx)$ over sentences. 
Let $\phi\in \FO(\bowtie)$ be a sentence, and let $k$ be the number of disjunctions and quantifiers in $\phi$. Let $\phi^*$ be obtained from $\phi$ by replacing all equiextension atoms of the form $\tuple x_1 \bowtie \tuple x_2$ with $\psi^k(\tuple x_1, \tuple x_2)$.   
We can make four simplifying assumption without loss of generality. 
First,  we may restrict attention to weighted semantics by item \ref{i:alternateprop} of Proposition \ref{propositions}. Thus, we assume that $\mA\models^w_{\X} \phi$ for some weighted team $\X$ and a finite structure $\mA$ with universe of size $n$.
Second, we may assume that the support of $\X$ consists of the empty assignment by item \ref{i:locality} of Proposition \ref{propositions}.
Third, since $\FO(\bowtie)$ is insensitive to assignment weights, we may assume that the satisfaction of $\phi$ by $\X$ is witnessed by uniform semantic operations. That is, existential and universal quantification split an assignment to at most $n$ equally weighted extensions, and disjunction can only split an assignment to two equally weighted parts. 
Fourth, we may assume that any equiextension atom $\tuple x_1 \bowtie \tuple x_2$
 appears in $\phi$ in an equivalent form $\exists uv(u\neq v \land \tuple x_1u \bowtie \tuple x_2v)$, to guarantee that the condition $|\X_{\tuple x_1 =\tuple x_2}|=0$ holds for all appropriate subteams $\X$.
%
We then obtain by the previous lemma and a simple inductive argument that $\mA\models^w_{\X} \phi^*$. The converse direction follows similarly by the previous lemma.
\end{proof}
Consequently, probabilistic inclusion logic captures $\PTIME$, for this holds already for inclusion logic \cite{gallhella13}. Another consequence is an alternative proof, through probabilistic inclusion logic (Theorem \ref{thm:sentprobinc}) and linear programs (Theorems \ref{thm:FOprob} and 
  \ref{thm:ptimedataesor}), for the PTIME upper bound of the data complexity of inclusion logic. For this, note also that quantification of functions, whose range is the unit interval, is clearly expressible in $\esor{\leq,\SUM,0,1}$.
\begin{corollary}\label{cor:pincP}
Sentences of $\FO(\approx)$ capture $\PTIME$ on finite ordered structures.
\end{corollary}


Theorem \ref{thm:sentprobinc} also extends to formulae over uniform teams. Recall that a function $f$ is uniform if $f(s)=f(s')$ for all $s,s'\in \supp{f}$.

\begin{theorem}
$\FO(\sub)\leq \FO(\approx)$ 
 over uniform probabilistic teams.
\end{theorem}
\begin{proof}
Recall that $\FO(\sub)\equiv \FO(\bowtie)$. Let $\phi$ be an $\FO(\bowtie)$ formula, $\A$ a finite structure, and $\X$ a uniform probabilistic team. Let ${}^*$ denote the translation of Theorem \ref{thm:sentprobinc}. Now
%
 \begin{linenomath*}
\begin{align*}
\mA \models_\X \phi \quad&\Leftrightarrow\quad (\mA, R\dfn X) \models \forall x_1\dots x_n \big(\neg R(x_1\dots x_n) \lor \big(R(x_1\dots x_n) \land \phi\big) \big)\\
&\Leftrightarrow\quad  (\mA, R\dfn X) \models \forall x_1\dots x_n \big(\neg R(x_1\dots x_n) \lor \big(R(x_1\dots x_n) \land \phi\big) \big)^*\\
&\Leftrightarrow\quad  (\mA, R\dfn X) \models \forall x_1\dots x_n \big(\neg R(x_1\dots x_n) \lor \big(R(x_1\dots x_n) \land \phi^*\big) \big)\\
&\Leftrightarrow\quad \mA \models_\X \phi^*,
\end{align*}
\end{linenomath*}
where $X$ is the support of $\X$ and $\Dom(\X)=\{x_1,\dots, x_n\}$.
   
\end{proof}

\iflong
\section{Definability over open formulae}\label{sect:sep}

We now turn to definability over open formulae. In team semantics, inclusion logic extended with dependence atoms is expressively equivalent to independence logic at the level of formulae. This relationship however does not extend to probabilistic team semantics. As we will prove next, probabilistic inclusion logic extended with dependence atoms is strictly less expressive than probabilistic independence logic. The reason, in short, is that logics with  marginal identity and dependence can only describe additive distribution properties, whereas the concept of independence involves multiplication. 

We begin with a proposition illustrating that probabilistic independence logic has access to irrational weights.\footnote{We thank Vadim Kulikov for the idea behind this proposition.}

\begin{proposition}\label{prop:esim}
Define $\phi(x)=\exists c\exists y\forall z\theta$, where $\theta$ is defined as
\begin{linenomath*}
\begin{equation}
\dep(c)\wedge \pmi{x}{y} \wedge x\approx y \wedge ((x=c \wedge y= c) \leftrightarrow z=c). \label{esimerkki}
\end{equation}
\end{linenomath*}
Let $\A$ be a finite structure with domain $A$ of size $n$, and let $\X$ be a probabilistic team. Then
\begin{equation}\label{eq:sqrt}
\A\models_\X \phi(x) \implies |\X_{x=a}|=\frac{1}{\sqrt{n}}\text{ for some }a\in A.
\end{equation}
\end{proposition}
\begin{proof}
Suppose  $\A\models_\X \phi(x)$, and let  $\Y$ be an extension of $\X$, in accord with the quantifier prefix of $\phi$, that satisfies \eqref{esimerkki}. Then in $\Y$ $c$ is constant and $z$ uniformly distributed over all domain values. Hence $z$ equals $c$ for weight $\frac{1}{n}$, and consequently $x$ and $y$ simultaneously equal $c$ for the same weight. Since $x$ and $y$ are independent and identically distributed, in isolation they equal $c$ for weight $\frac{1}{\sqrt{n}}$. Since $\X$ and $\Y$ agree on the weights of $x$, the claim follows.    
\end{proof}

It follows, then, that independence atoms are not definable in additive existential second-order logic.
\begin{lemma}
$\FO(\pind) \not\leq \esor{\leq,+,0,1}$.
\end{lemma}
\begin{proof}
Let $\phi(x)$ be as in the previous proposition.
Assume towards contradiction that it has a translation $\Psi(f)$ in $\esor{\leq,+,0,1}$. Then $\Psi$ contains one free unary function variable $f$ to encode the probabilistic team over $\{x\}$. 
Let $\A$ be a structure with universe $\{0,1\}$ and empty vocabulary.  By the previous proposition $\A$ satisfies $\Psi(f)$ if and only if $\{f(0),f(1)\}=\{1/\sqrt{2}, 1-1/\sqrt{2}\}$.  

We define a translation $\Phi \mapsto \Phi{}^*$ from $\esor{\leq,+,0,1}$ over $\A$ to the additive existential (first-order) theory over the reals. Without loss of generality $\Phi$ has no nested function terms. In the translation, we interpret function terms of the form $g(a_1, \ldots ,a_{\ar(g)})$, for $a_1, \ldots , a_{\ar(g)} \in \{0,1\}$, as first-order variables.
 The translation, defined recursively, is identity for numerical inequality atoms, 
 homomorphic for disjunction and conjunction, and otherwise defined as: 
\begin{itemize}
\item $(\forall y \Phi)^*:=\Phi^*(0/y) \land \Phi^*(1/y)$,
\item $(\exists y \Phi)^*:=\Psi^*(0/y) \lor \Phi^*(1/y)$,
\item $(\exists g \Phi)^*:= (\exists g(a_1, \ldots ,a_{\ar(g)})_{a_1, \ldots , a_{\ar(g)} \in \{0,1\}}\Phi^*$,
\end{itemize}
where  $\Phi^*(a/y)$ is obtained from $\Phi^*$ by substituting variable $h(x_1, \ldots ,x_{i-1}, a, x_{i+1}, \ldots x_n)$ for any variable of the form  $h(x_1, \ldots ,x_{i-1}, y, x_{i+1}, \ldots x_n)$. Applying the translation to $\Psi(f)$ we obtain a formula $\Psi^*(f(0),f(1))$ that contains two free first-order variables $f(0)$ and $f(1)$.

It is easy to see that $\A\models \Psi(f)$ if and only if $\Psi^*(f(0),f(1))$ holds in the real arithmetic. Consequently, $\Psi^*$ has only irrational solutions.  On the other hand, $\Psi^*$ can be transformed to the form $\exists x_1 \ldots \exists x_n \bigvee_i \bigwedge_j C_{ij}$, where each $C_{ij}$ is a (strict or non-strict) linear inequation with integer coefficients and constants. Since $\Psi^*$ is satisfiable, some system of linear inequations $\bigwedge_j C_{ij}$ has solutions, and thus also rational solutions.
\footnote{
To see why, observe that such a system can be expressed as a linear program in the canonical form (e.g., as in the proof of Theorem \ref{addboole}). Since the optimal solution of a linear program is always attained at a vertex of the feasible region, a linear program with rational coefficients and constants has at least one rational optimal solution if it has optimal solutions at all (see, e.g., \cite{dantzig97}).
}
 Thus $\Psi^*$ has rational solutions, which leads to a contradiction.  We conclude that $\phi(x)$ does not translate into $\esor{\leq,+,0,1}$.    
\end{proof}

The following result is now immediate. 
\begin{theorem}\label{thm:sep}
$\FO(\deps,\approx)< \FO(\pind)$.
\end{theorem}
\begin{proof}
Dependence and marginal identity atoms are definable in $\FO(\pind)$ (i.e., in first-order logic extended with marginal probabilistic independence atoms) \cite[Proposition 3, Theorem 10, and Theorem 11]{HHKKV19}. Furthermore, $\phi(x)$ in Proposition \ref{prop:esim} is not definable in $\FO(\deps,\approx)$. For this, recall that by Theorem \ref{thm:FOprob},
$\FO(\deps,\approx)$ corresponds to $\pesou{\leq,+,0,1}$. This logic is clearly subsumed by $\esor{\leq,+,0,1}$, which in turn cannot translate $\phi(x)$ by the previous lemma.
    \end{proof}

There are, in fact, more than one way to prove that $\FO(\pind) \not \leq \FO(\deps,\approx)$. 
 Above, we use the fact that probabilistic independence cannot be defined in terms of additive existential second-order logic, which in turn encompasses both dependence and marginal independence atoms. Another strategy is to apply the closure properties of these atoms.

Let $\phi$ be a formula over probabilistic team semantics.
We say that $\phi$ is \emph{closed under scaled unions} if for all parameters $\alpha \in [0,1]$, finite structures $\A$, and probabilistic teams $\X$ and $\Y$: $\A\models_\X \phi$ and $\A\models_\Y \phi$ imply $\A \models_\Z \phi$, where $\Z\dfn\alpha \X + (1-\alpha)\Y$. In the weighted semantics, we say that $\phi$ is \emph{closed under unions} if for all finite structures $\A$ and weighted teams $\X$ and $\Y$: $\A\models^w_\X \phi$ and $\A\models^w_\Y \phi$ imply $\A \models^w_{\X+\Y} \phi$.
We say that $\phi$ is \emph{relational} if for all finite structures $\A$, and probabilistic teams $\X$ and $\Y$ such that $\supp{Y} = \supp{X}$: $\A\models_\X \phi$ if and only if $\A\models_\Y \phi$. We say that $\phi$ is \emph{downwards closed} if for all finite structures $\A$, and probabilistic teams $\X$ and $\Y$ such that $\supp{Y} \subseteq \supp{X}$: $\A\models_\X \phi$ implies $\A\models_\Y \phi$. Furthermore, a logic $\calL$ is called relational (downward closed, closed under scaled union, resp.) if each formula $\phi$ in $\calL$ is relational (downward closed, closed under scaled unions, resp.).
%
\begin{proposition}\label{thm:closure}
The following properties hold:
\begin{itemize}
\item $\FO(\deps)$ is relational. [Self-evident]
\item $\FO(\approx)$ is closed under scaled unions. \cite{HHKKV19}
\end{itemize}
\end{proposition}

In the context of multiteam semantics, Gr\"adel and Wilke have shown that probabilistic independence is not definable by any logic that extends first-order logic with a collection of atoms that are downwards closed or union closed \cite[Theorem 17]{wilke20}. In fact, their proof works also when downwards closed atoms are replaced with relational atoms (which, in their framework as well as in the probabilistic framework, is a strictly more general notion). While their proof technique does not directly generalise to probabilistic team semantics, it can readily be adapted to weighted semantics (Definition \ref{def:weightedsemantics}).
\begin{theorem}[cf. \cite{wilke20}]\label{thm:wilke}
Let $\calC$ be a collection of relational atoms, and let $\calD$ be a collection of atoms that are closed under unions. Then under weighted semantics
$\FO(\pind) \not\leq \FO(\calC,\calD)$.
\end{theorem}

This theorem can be then transferred to probabilistic semantics by using the following observations: For any probabilistic $n$-ary atom ${\rm D}$, we can define an $n$-ary atom ${\rm D^*}$ in the weighted semantics as follows:
\[
\mA \models^w_{\X} {\rm D^*}(x_1,\dots,x_n) \text{ if and only if } \mA \models_{\frac{1}{\lvert \X \rvert}\cdot\X} {\rm D}(x_1,\dots,x_n) 
\] 
It follows via a straightforward calculation that ${\rm D^*}$ is union closed, whenever ${\rm D}$ is closed under scaled unions: Assume that $\mA \models^w_{\X} {\rm D^*}(x_1,\dots,x_n)$ and $\mA \models^w_{\Y} {\rm D^*}(x_1,\dots,x_n)$. Fix $k=\frac{\lvert \X \rvert}{\lvert \X \rvert + \lvert \Y \rvert}$ and note that then $1-k = \frac{\lvert \Y \rvert}{\lvert \X \rvert + \lvert \Y \rvert}$. By definition, we get $\mA \models_{\frac{1}{\lvert \X \rvert}\cdot\X} {\rm D}(x_1,\dots,x_n)$ and $\mA \models_{\frac{1}{\lvert \Y \rvert}\cdot\Y} {\rm D}(x_1,\dots,x_n)$, from which $\mA \models_{\frac{k}{\lvert \X \rvert}\cdot\X + \frac{1-k}{\lvert \Y \rvert}\cdot\Y} {\rm D}(x_1,\dots,x_n)$ follows via closure under scaled unions. Finally, since $\frac{k}{\lvert \X \rvert}\cdot\X + \frac{1-k}{\lvert \Y \rvert}\cdot\Y =  \frac{1}{\lvert \X \rvert+\lvert \Y \rvert}\cdot\X + \frac{1}{\lvert \X \rvert+\lvert \Y \rvert}\cdot\Y = \frac{1}{\lvert \X \rvert+\lvert \Y \rvert}\cdot(\X+\Y)$, we obtain that $\mA \models^w_{\X+\Y} {\rm D^*}(x_1,\dots,x_n)$.

The final piece of the puzzle is the following generalisation of \cite[Proposition 8]{HHKKV19}. The original proposition was formulated for concrete atomic dependency statements satisfying the proposition as an atomic case for induction. The inductive argument of the original proof works with any collection of atoms that satisfy the proposition as an atomic case.
\begin{proposition}\label{prop:jelia19}
Let $\D$ be a collection of atoms. 
  If
  \(
  \A\models^w_\X {\rm D}(\vec{x}) \Leftrightarrow
  \A\models_{\frac{1}{|\X|} \cdot \X} {\rm D}(\vec{x}),
  \)
  for every structure $\A$, weighted team $\X:X\to \mathbb{R}_{\geq 0}$  of $\A$, and $D\in\D$, then
  \(
  \A\models^w_\X \phi \Leftrightarrow
  \A\models_{\frac{1}{|\X|} \cdot \X} \phi,
  \)
   for every $\A$, $\X$, and $\phi \in \FO(\D)$ as well.
\end{proposition}

By combining Theorem \ref{thm:wilke} and Proposition \ref{prop:jelia19} with the two observation made above, we obtain the probabilistic analogue of Theorem \ref{thm:wilke}.
\begin{theorem}
Let $\calC$ be a collection of relational atoms, and let $\calD$ be a collection of atoms that are closed under scaled unions. Then
$\FO(\pind) \not\leq \FO(\calC,\calD)$.
\end{theorem}

From this, $\FO(\pind) \not \leq \FO(\deps,\approx)$ follows as a special case by Proposition \ref{thm:closure}.

\section{Axiomatization of marginal identity atoms}\label{sect:ax}
Next we turn to axioms of the marginal identity atom, restricting attention to atoms of the form $x_1\ldots x_n\approx y_1\ldots y_n$, where both $x_1\ldots x_n$ and $y_1\ldots y_n$ are sequences of distinct variables. It turns out that the axioms of inclusion dependencies over relational databases \cite{CasanovaFP84} are sound and almost complete for marginal identity; we only need one additional rule for symmetricity.  Consider the following axiomatization:

\begin{enumerate}
\item reflexivity: $x_1\ldots x_n\approx x_1\ldots x_n$;
\item symmetry: if $x_1\ldots x_n\approx y_1\ldots y_n$, then $y_1\ldots y_n\approx x_1\ldots x_n$;
\item projection and permutation:  if $x_1\ldots x_n\approx y_1\ldots y_n$, then $x_{i_1}\ldots x_{i_k}\approx y_{i_1}\ldots y_{i_k}$, where $i_1,\ldots ,i_k$ is a sequence of distinct integers from $\{1, \ldots ,n\}$.
\item transitivity: if $x_1\ldots x_n\approx y_1\ldots y_n$ and $y_1\ldots y_n\approx z_1\ldots z_n$, then $x_1\ldots x_n\approx z_1\ldots z_n$.
\end{enumerate}

For a set of marginal identity atoms $\Sigma\cup\{\sigma\}$, a proof of $\sigma$ from $\Sigma$ 
 is a finite sequence of marginal identity atoms such that (i) each element of the sequence is either from $\Sigma$, or follows from previous atoms in the sequence by an application of a rule, and (ii) the last element in the sequence is $\sigma$. We write $\Sigma \vdash \sigma$ if there is a proof of $\sigma$ from $\Sigma$.
 For a probabilistic team $\X$ and a formula $\phi$ over the empty vocabulary $\tau_\emptyset$, we write $\X \models \phi$ as a shorthand for $\mA\models_\X \phi$, where $\A$ is the structure over $\tau_\emptyset$ whose domain consists of the values in the support of $\X$. We use a shorthand $X\models \phi$, for a team $X$, analogously.
 We write $\Sigma \models \sigma$ if every probabilistic team $\Y$ that satisfies $\Sigma$ satisfies also $\sigma$.
  The proof of the following theorem is an adaptation of a similar result for inclusion dependencies \cite{CasanovaFP84}.
 
\begin{theorem}
Let $\Sigma\cup\{\sigma\}$ be a finite set of marginal identity atoms. Then $\Sigma\models \sigma$ if and only if $\Sigma\vdash \sigma$.
\end{theorem}
\begin{proof}
It is clear that the axiomatization is sound; we show that it is also complete. 

Assume that $\Sigma\models \sigma$, where $\sigma$ is of the form $  x_1\ldots x_n\approx y_1\ldots y_n$. 
 Let $\base$ consist of the variables appearing in $\Sigma\cup \{\sigma\}$. For each subset $V\sub \base$, let $i_V$ be an auxiliary variable, called an \emph{index}. Denote the set of all indices over subsets of $\base$ by $\idx$. Define $\Sigma^*$ as the set of all inclusion atoms $u_1\ldots u_l i_U\sub v_1\ldots v_l  i_V$, where $U=\{ u_1,\ldots ,u_l\}$, $V=\{v_1, \ldots ,v_l\}$, and $u_1\ldots v_l\approx v_1\ldots v_l$ or its inverse $v_1\ldots v_l \approx u_1\ldots v_l$ is in $\Sigma$.

To show that $\Sigma\vdash x_1\ldots x_n\approx y_1\ldots y_n$, we will first apply the chase algorithm of database theory to obtain a finite team $Y$ that satisfies $\Sigma^*$, where the codomain of $Y$ consists of natural numbers.
The indices $i_V$ in $Y$, in particular, act as multiplicity measures for values of $V$, making sure that both sides of any marginal identity atom in $\Sigma$ appear in $Y$ with equal frequency. This way, the probabilistic team $\Y$, defined as the uniform distribution over $Y$, will in turn satisfy $\Sigma$. 
Finally, we show that the chase algorithm yields a proof of $\sigma$, utilizing the fact that $\Y$ satisfies $\sigma$ by assumption.

  Next, we define a team $X_0$ that serves as the starting point of the chase algorithm. We also describe how assignments over $\base$ that are introduced during the chase are extended to $\base \cup \idx$.
  
  Let $X_0=\{s^*\}$, where $s^*$ is an assignment defined as follows.
   Let $s^*(x_i)=i$, for $1 \leq i \leq n$, and $s^*(x)=0$, for  $x\in (\base \cup \idx) \setminus \{x_1, \ldots ,x_n\}$. 
 For a team $Y$ with variable domain $\base \cup \idx$ 
 and an assignment $s$ with variable domain $\base$, define $s_Y\colon \base \cup \idx \to \mathbb{N}$ as the extension of $s$
  such that 
 \begin{equation}\label{eq:const}
 s_Y(i_V) = | \{t \in Y \mid t \upharpoonright V = s \upharpoonright V\}|,
 \end{equation}
 for $i_V \in \idx$.
 That is, the value $s_Y(i_V)$ 
  is the number of repetitions of $s \upharpoonright V$ in $Y$.

 In what follows, we describe a chase rule to expand  a team $X$. We say that an assignment $s'$ \emph{witnesses} an inclusion atom $\tuple x \sub \tuple y$ for another assignment $s$, if $s(\tuple x) = s'(\tuple y)$.
Consider  the following chase rule:

\paragraph{Chase rule}
Let $X$ be a team with variable domain $\base \cup \idx$, $s\in X$, and $\sigma:= u_1\ldots u_l i_U\sub v_1\ldots v_l  i_V\in \Sigma^*$.
Suppose no assignment in $X$  witnesses $ \sigma$ for $s$. Now let $s'$ be the assignment with variable domain $\base$ that is defined as 
\[
s'(x):=\begin{cases}
s(u_j) &\text{if $x$ is $v_j$, and}\\
0 &\text{ otherwise.}\\
\end{cases}
\]
Then we say that $s$ and $\sigma$ \emph{generate} the assignment $s'_{X}$. 

Next, let  $\calS=(X_0,X_1, X_2, \ldots) $ be a maximal sequence, where $X_{i+1}=X_i\cup\{s'_{X_i}\}$ for an assignment $s'_{X_i}$ generated non-deterministically 
 by some $s\in X_j$ and $\tau\in\Sigma^*$ according to the chase rule, where $j\leq i$ is minimal. Define $Y$ as the union of all elements in $\calS$. Note that $Y$ is finite if $\calS$ is. In particular, if $Y$ is finite, then it equals $X_i$, where $i$ is the least integer such that the chase rule is not anymore applicable to $X_i$. Below, we will show that $Y$ is finite, which follows if the chase algorithm terminates.
 

It is easy the verify that the following holds, for each $i\in \N$: For any  $U=\{u_1,\dots,u_n\}$ and $s\in X_i$, if the team $X_s:=\{t \in X_i \mid t \upharpoonright U = s \upharpoonright U\}$ is of size $m$, 
 then $\{t(i_U)\mid t\in X_s\}=\{0,\dots ,m-1\}$. That is, the values of $i_U$ in $X_s$ form an initial segment of $\mathbb{N}$ of size $\lvert X_s \rvert $.
Therefore, if $s\in X_i$ has no witness for $u_1\ldots u_l i_U\sub v_1\ldots v_l  i_V$ in $X_i$, then for any $t\in X_i$ such that $s(u_1\ldots u_l)= t(v_1\ldots v_l)$, we have $s(i_U)> t(  i_V)$.
It follows that 
\begin{equation}\label{eq:prop}
\text{
$ s( i_U)\geq s'_{X_i}(i_V) $ if $s'_{X_i}$ is generated by $s\in X_i$ and $u_1\ldots u_l i_U\sub v_1\ldots v_l  i_V$.
}
\end{equation} 

%
%

We will next show how $\Sigma\vdash x_1\ldots x_n\approx y_1\ldots y_n$ follows from the following two claims. We will then prove the claims, which concludes the proof of the theorem. 
\begin{claim}\label{claim:one}
$Y$ is finite.
\end{claim}

\begin{claim}\label{claim:two}
If $Y$ contains an assignment $s$ that maps some sequence of variables $z_j$, for $1 \leq j \leq k$, to distinct  $1\leq i_j\leq n$, then $\Sigma\vdash x_{i_1}\ldots x_{i_k}\approx z_1\ldots z_k$.
\end{claim}

It follows by construction that $Y\models \Sigma^*$. Since $Y$ is finite by Claim \ref{claim:one}, we may define a probabilistic team $\Y$ as the uniform distribution over $Y$. By the construction of $Y$ and $\Sigma^*$, it follows that $\Y\models \Sigma$, and hence $\Y \models x_1\ldots x_n\approx y_1\ldots y_n$ follows from the assumption that $\Sigma\models \sigma$. Consequently, $Y$ contains an assignment $s$ which maps $y_i$ to $i$, for $1 \leq i\leq n$. We conclude that by Claim \ref{claim:two} there is a proof of $  x_1\ldots x_n\approx y_1\ldots y_n$ from $\Sigma$. 
 \footnote{
 Claim \ref{claim:two}  is essentially from  \cite{CasanovaFP84}, with the exception that here we also need to consider symmetricity. This claim intuitively states that the chase procedure produces only assignments whose corresponding marginal identity atoms are provable from $\Sigma$
 }
 
 To complete the proof, we prove Claims \ref{claim:one} and \ref{claim:two}.

\begin{proof}[Proof of Claim \ref{claim:one}]
Assume towards contradiction that $Y$ is infinite, which entails that the sequence $\calS=(X_0,X_1, X_2, \ldots )$ is infinite.
 W.l.o.g. the chase rule is always applied to $s$ that belongs to the intersection $X_i \cap X_{j}$, for minimal $j\leq i$. 
 Define $\calS'=(X'_{0}, X'_{1}, X'_{2}, \ldots) $ as the sequence, where $X'_0=X_0$, and $X'_{i+1}$ is defined as $X_j$ where $j$ is the least integer such that  all $s\in X'_{i}$ and $\sigma\in \Sigma^*$ have a witness in $X_j$.
Due to the application order of the chase rule, it follows that 
\begin{equation}\label{eq:teksti}
\text{any assignment in $X'_{i+1}\setminus X'_{i}$ is generated by some assignment in $X'_{i}\setminus X'_{i-1}$,
}
\end{equation}
assuming $X'_{-1}=\emptyset$. Moreover, $\calS'$ is a subsequence of $\calS$ which is finite iff $\calS$ is.

 We first define some auxiliary concepts. 
 For an assignment $s$ in $X$, we use a shorthand $\bsupp{s}$ for $s\upharpoonright \base$,  called the \emph{base} of $s$. We also define $\bsupp{X}\dfn\{\bsupp{s}\mid s\in X\}$. 
 The \emph{multiplicity in $X$} of an assignment $s$ is defined as  
  $ |\{s'\in X\mid \bsupp{s} =\bsupp{s'}\}|$. 
 Note that $\bsupp{Y}$ is finite, for $\bsupp{s}$ is a mapping from $\base$ into $\{0, \ldots ,n\}$ for all $s\in Y$.
 Thus, since $Y$ is infinite, it contains assignments with infinite multiplicity in $Y$. 
 Next, we associate each assignment $s$ with the set of its \emph{positive variables} $\positive{s}\dfn\{x\in \base \mid s(x)> 0\}$, the size of which is called the  \emph{degree} of $s$.

Let $k$ be some integer such that $X'_{k}$  contains every assignment in $Y$ that has finite multiplicity in $Y$, and denote $X'_k$ by $Z$. 
Let $M\in \{1, \ldots ,n\}$ be the maximal degree of any assignment in $Y$ with infinite multiplicity in $Y$, that is, the maximal degree of any assignment in $Y\setminus Z$. 
Then, take any $s_L\in X'_{L} \setminus X'_{L-1}$ of degree $M$, where $L> k+S$ for $S\dfn |\bsupp{Y}|$. By  property \eqref{eq:teksti}, we find a sequence of assignments $(s_0, \ldots ,s_L)$, where $s_{i+1} \in X'_{i+1}\setminus X'_{i}$, for $i< L$, was generated by $s_{i}\in X'_{i}\setminus X'_{i-1} $ with the chase rule. 
  Since $S$ is sufficiently large, this sequence has a suffix $(s_l, \ldots , s_{m}, \ldots ,s_{L})$ in which each assignment belongs to $Y\setminus Z$, has degree $M$, and where $l < m$ and $\bsupp{s_l} = \bsupp{s_{m}}$.

It now suffices to show the following subclaim:
\begin{subclaim}
If $t,t'\in Y\setminus Z$ are two assignments with degree $M$ such that $t'$ was generated by $t$ by the chase rule, then
 $t(i_{\positive{t}}) \geq t'(i_{\positive{t'}}) $. 
\end{subclaim}

 The subclaim implies that $s_l(i_{\positive{s_l}})\geq s_m(i_{\positive{s_m}})$, which leads to a contradiction. For this, observe that the assignment construction in \eqref{eq:const},  together with $\bsupp{s_l} = \bsupp{s_{m}}$, implies that $s_l(i)< s_m(i)$ for all indices $i$. In particular, we have  
 $s_l(i_{\positive{s_l}}) < s_m(i_{\positive{s_m}}) $ since $\positive{s_l}=\positive{s_m}$. Hence, the assumption that $Y$ is infinite must be false.
\end{proof}

\begin{proof}[Proof of the subclaim]
%
%
%
 Suppose $t'$ is generated by $t$ and $u_1\ldots u_l i_U\sub v_1\ldots v_l  i_V\in \Sigma^*$. Without loss of generality  ${\positive{t}} =\{u_1, \ldots ,u_M\}$, in which case ${\positive{t'}} = \{v_1, \ldots ,v_M\}$.
We need to show that $t(i_{\positive{t}} ) \geq t'(i_{\positive{t'}})$.
  Now, $(t(u_1), \ldots ,t(u_l))$ is a sequence of the form $(i_1, \ldots ,i_M,0\ldots ,0)$, where $i_j$ are positive integers. By the assumption that $t\in Y\setminus Z$, there is an integer $m$ such that $t\in X_{m+1}\setminus X_{m}$ and $Z \subseteq X_m$.
 We obtain that
\begin{align}
t(i_{\positive{t}})
=\,&  \lvert \{s \in X_m \mid (s(u_1), \ldots ,s(u_M)) = (i_1, \ldots ,i_M)\} \rvert\label{nolla}\\
=\,&\sum_{j_{M+1}, \ldots ,j_l\in\{0,\ldots ,n\}} 
|\{s \in X_m \mid (s(u_1), \ldots ,s(u_l)) = (i_1, \ldots ,i_M,j_{M+1},\ldots ,j_l)\}|\nonumber\\
=\, &
|\{s \in X_m \mid (s(u_1), \ldots ,s(u_l)) = (i_1, \ldots ,i_M,0\ldots ,0)\}|
+\nonumber\\
&\sum_{\substack{j_{M+1}, \ldots ,j_l\in\{0,\ldots ,n\}\nonumber\\
(j_{M+1}, \ldots ,j_l)\neq (0, \ldots ,0) }} 
|\{s \in X_m \mid (s(u_1), \ldots ,s(u_l)) = (i_1, \ldots ,i_M,j_{M+1}\ldots ,j_l)\}|\nonumber\\
=\, &
t(i_U)
+
\sum_{\substack{j_{M+1}, \ldots ,j_l\in\{0,\ldots ,n\}\\(j_{M+1}, \ldots ,j_l)\neq (0, \ldots ,0) }} 
|\{s \in Z \mid (s(u_1), \ldots ,s(u_l)) = (i_1, \ldots ,i_M,j_{M+1},\ldots ,j_l)\}|\label{eka}\\
\geq\, &
t'(i_V)
+
\sum_{\substack{j_{M+1}, \ldots ,j_l\in\{0,\ldots ,n\}\\(j_{M+1}, \ldots ,j_l)\neq (0, \ldots ,0) }} 
|\{s \in Z \mid (s(v_1), \ldots ,s(v_l)) = (i_1, \ldots ,i_M,j_{M+1},\ldots ,j_l)\}|\label{toka}\\
=\, &
t'(i_{\positive{t'}})\nonumber
\end{align}
Here, the assignment construction in \eqref{eq:const} entails \eqref{nolla}, and it is also used in \eqref{eka}. For the summation term appearing in \eqref{eka}, 
we note that each assignment whose degree is strictly greater than $M$ must belong to $Z$. It remains to consider \eqref{toka}; the last equality is symmetrical to the composition of the first four equalities.

To show that \eqref{toka} holds, observe first that $t(i_U)\geq t'(i_V)$ by property \eqref{eq:prop}. For the summation term appearing in \eqref{toka}, suppose $\alpha=|\{s \in Z \mid (s(v_1), \ldots ,s(v_l)) = (i_1, \ldots ,i_M,j_{M+1},\ldots ,j_l)\}|$, for some sequence $j_{M+1}, \ldots ,j_l\in \{0, \ldots ,n\}$ containing a positive integer. By the assignment construction in \eqref{eq:const}, we find an assignment $s_0\in Z$ such that $(s_0(v_1), \ldots ,s_0(v_l),s_0(i_V)) = (i_1, \ldots ,i_M,j_{M+1},\ldots ,j_l,\alpha-1)$. Observe that $v_1\ldots v_l i_V\sub u_1\ldots u_l  i_U\in \Sigma^*$, because $\Sigma^*$ is symmetrical.
Now, since $Z$ is subsumed by $Y$, which in turn satisfies $v_1\ldots v_l i_V\sub u_1\ldots u_l  i_U$, we  find an assignment $s_1\in Y$ such that $(s_1(u_1), \ldots ,s_1(u_l),s_1(i_U)) = (i_1, \ldots ,i_M,j_{M+1},\ldots ,j_l,\alpha-1)$. Since the degree of $s_1$ is greater than $M$, we observe that $s_1\in Z$. This entails that $\alpha\leq |\{s \in Z \mid (s(u_1), \ldots ,s(u_l)) = (i_1, \ldots ,i_M,j_{M+1},\ldots ,j_l)\}|$ by the assignment construction in \eqref{eq:const}. 
From this, we obtain that \eqref{toka} holds. This shows the subclaim.
\end{proof}

\begin{proof}[Proof of Claim \ref{claim:two}]
Note that, if $s\in Y$, then there exists a minimal $i$ such that $s\in X_i\setminus X_{i-1}$.
We prove the claim by induction on $i$. For the initial team $X_0 = \{s^*\}$, we have $s^*(x_i)=i$, for $1 \leq i\leq n$. By reflexivity we obtain $x_{i_1}\ldots x_{i_k}\approx x_{i_1}\ldots x_{i_k}$, and thus the claim holds for the base step.

For the inductive step, suppose $s\in X_{i+1}\setminus X_{i}$ is generated by some $s'\in X_{j}\setminus X_{j-1}$, $j\leq i$, and some $u_1\ldots u_l i_U\sub v_1\ldots v_l  i_V$ in $\Sigma^*$. For a variable $v_i$ from $v_1,\ldots, v_l $ we say the variable $u_i$ from $u_1, \ldots ,u_l$ is its \emph{corresponding} variable.
Let $z_1, \ldots ,z_k$ be variables as in the claim, i.e., $s(z_j)=i_j\geq 1$, for $1 \leq j\leq k$. Now from the construction of $s$ (i.e., \eqref{eq:const}) it follows that $z_1, \ldots ,z_k$ are variables from $v_1,\ldots, v_l$. Let $z'_1, \ldots ,z'_k$ from $u_1, \ldots ,u_l$ denote the corresponding variables of $z_1, \ldots ,z_k$.  
Since $s$ was constructed by $s'$ and $u_1\ldots u_l i_U\sub v_1\ldots v_l  i_V$, it follows that $s(z_1, \ldots ,z_k)= s'(z'_1, \ldots ,z'_k)$.
By applying the induction hypothesis to $s'$, we obtain that $\Sigma$ yields a proof of $x_{i_1}\ldots x_{i_k} \approx z'_1\ldots z'_k$.
Since $u_1\ldots u_l \approx v_1\ldots v_l $ or its inverse is in $\Sigma$, using projection and permutation (and possibly symmetricity) we can deduce $z'_1\ldots z'_k\approx z_1\ldots z_k$.
Thus by transitivity we obtain a proof of $x_{i_1}\ldots x_{i_k}  \approx z_1\ldots z_k$. This concludes the proof of the claim. 
\end{proof}
\end{proof}

\fi

\section{Conclusion}
Our investigations gave rise to the expressiveness hierarchy in Table \ref{tbl:results}.
\begin{table}[t]
	\begin{tabular}{cccccccc}	
		&&$\text{ almost conjunctive }\leau{=,\SUM,0,1}$& &$\pesou{=,+,0,1}$&&$\pesou{=,\times,+,0,1}$ \\
		&& \rotatebox[origin=c]{-90}{$\equiv$}${}^*$ &&\rotatebox[origin=c]{-90}{$\equiv$}${}^*$ && \rotatebox[origin=c]{-90}{$\equiv$}$\stackrel{\scriptsize{\text{\cite{abs-2003-00644}}}}{}$ \\
		&&$\FO(\approx)$&$<\stackrel{\scriptsize{\text{\cite{HHKKV19}}}}{}$ &$\FO(\approx, \deps)$&$<^*$&$\FO(\pind)$\\
		&&  && &&\rotatebox[origin=c]{-90}{$\equiv$}$\stackrel{\scriptsize{\text{\cite{HHKKV19}}}}{}$ \\
		&&&&&& $\FO(\cpind)$ 
	\end{tabular}
\caption{
The known expressivity hierarchy of logics with probabilistic team semantics and corresponding $\ESO$ variants on metafinite structures.  
The results of this paper are marked with an asterisk (*).  
}
\label{tbl:results}
\end{table}
Furthermore, we established that $\FO(\approx)$ captures $\PTIME$ on finite ordered structures, and that $\FO(\approx,\deps)$ captures $\NP$ on finite structures.
Its worth to note that almost conjunctive $\eear{\leq ,+,\SUM,0,1}$ is in some regard a maximal tractable fragment of additive existential second-order logic, as dropping either the requirement of being almost conjunctive, or that of having the prefix form $\existst^*\exists^*\forall^*$, leads to a fragment that captures $\NP$. We also showed that the full additive existential second-order logic (with inequality and constants $0$ and $1$) collapses to $\NP$, a result which as far as we know has not been stated previously. 


Lastly, extending the axiom system of inclusion dependencies with a symmetry rule, we presented a sound and complete axiomatization for marginal identity atoms. 
Beside this result, it is well known that also marginal independence has a sound and complete axiomatization \cite{geiger:1991}. These two notions play a central role in statistics, as it is a common assumption in hypothesis testing that samples drawn from a population are independent and identically distributed (i.i.d.). 
It is an interesting open question whether marginal independence and marginal identity, now known to be axiomatizable in isolation, can also be axiomatized together. 


\section{Acknowledgements}
We would like to thank the anonymous referee for a number of useful suggestions. We also thank Joni Puljuj\"{a}rvi and Richard Wilke for pointing out errors in the previous manuscripts.

\vspace{-2mm}

\bibliographystyle{plain}
\bibliography{biblio,biblio2}


\appendix
\section{BSS-toolbox}\label{sect:app}
In this section we give a short introduction to BSS machines (see e.g. \cite{BSSbook}). The inputs for BSS machines come from $\RE^* \dfn \bigcup \{\RE^n \mid n\in\N\}$, which can be viewed as the real analogue of $\Sigma^*$ for a finite set $\Sigma$. The \emph{size} $|x|$ 
of $x\in \RE^n$ is defined as $n$. We also define $\RE_*$ as the set of all sequences $x=(x_i)_{i\in \mathbb{Z}}$ where $x_i\in \RE$. The members of $\RE_*$ are thus bi-infinite sequence of the form $(\ldots, x_{-2},x_{-1},x_0,x_1,x_2,\ldots )$. Given an element $x\in \RE^* \cup \RE_*$ we write $x_i$ for the $i$th coordinate of $x$. The space $\RE_*$ has natural shift operations. We define shift left $\sigma_l\colon \RE_* \to \RE_*$ and shift right $\sigma_r\colon\RE_* \to \RE_*$ as $\sigma_l(x)_i\dfn  x_{i+1}$ and $\sigma_r(x)_i\dfn x_{i-1}$.

\begin{definition}[BSS machines] \label{def:BSS} A BSS machine consists of an input space $\mathcal{I}=\RE^*$, a state space $\mathcal{S}=\RE_*$, and an output space $\mathcal{O}=\RE^*$, together with a connected directed graph whose nodes are labelled by $1, \ldots ,N$. The nodes are of five different types.
\begin{enumerate}
\item \emph{Input node}. The node labeled by $1$ is the only input node. The node is associated with a next node $\beta(1)$ and the input mapping $g_I: \mathcal{I} \to \mathcal{S}$.
\item \emph{Output node}. The node labeled by $N$ is the only output node. This node is not associated with any next node. Once this node is reached, the computation halts, and the result of the computation is placed on the output space by the output mapping $g_O: \mathcal{S}\to \mathcal{O}$.
\item \emph{Computation nodes.} A computation node $m$ is associated with a next node $\beta(m)$ and a mapping $g_m: \mathcal{S}\to \mathcal{S}$ such that for some $c\in \mathbb{R}$ and $i,j,k\in \mathbb{Z}$ the mapping $g_m$ is identity on coordinates $l\neq i$ and on coordinate $i$ one of the following holds:
\begin{itemize}
\item $g_m(x)_i =x_j+x_k$ (addition),
\item $g_m(x)_i = x_j - x_k$ (subtraction),
\item $g_m(x)_i = x_j\times x_k$ (multiplication),
\item $g_m(x)_i = c$ (constant assignment).
\end{itemize} 
\item \emph{Branch nodes.} A branch node $m$ is associated with nodes $\beta^-(m)$ and $\beta^+(m)$. Given $x\in \mathcal{S}$ the next node is $\beta^-(m)$ if $x_0\leq 0$, and $\beta^+(m)$ otherwise.
\item \emph{Shift nodes.} A shift node $m$ is associated either with shift left $\sigma_l$ or shift right $\sigma_r$, and a next node $\beta(m)$.
\end{enumerate}
The input mapping $g_I: \mathcal{I} \to \mathcal{S}$ places an
input $(x_1, \ldots ,x_n)$ in the state
\[(\ldots ,0,n,x_1, \ldots ,x_n, 0,\ldots )\in \mathcal{S},\]
 where the size of the input $n$ is located at the zeroth
 coordinate. The output mapping $g_O\colon \mathcal{S}\to
 \mathcal{O}$ maps a state to the string consisting of its first
 $l$ positive coordinates, where $l$ is the number of consecutive ones stored in the negative coordinates starting from the first negative coordinate.
 For instance, $g_O$ maps 
\[(\ldots ,2,1,1,1,n,x_1, x_2,x_3,x_4,\ldots )\in \mathcal{S},\]
 to $(x_1, x_2,x_3)\in \mathcal{O}$.
 A configuration at any moment of computation consists of a node
 $m\in \{1, \ldots ,N\}$ and a current state $x\in\mathcal{S}$.
The (sometimes partial) \emph{input-output} function $f_M:\RE^*\to \RE^*$ of a
machine $M$ is now defined in the obvious manner.
A function $f:\RE^*\to \RE^*$ is \emph{computable} if $f=f_M$ for some machine $M$. A language $L\sub \RE^*$ is \emph{decided} by a BSS machine $M$ if its characteristic function $\chi_L\colon \RE^*\to \RE^*$ is $f_M$.
 \end{definition}

\paragraph{Deterministic complexity classes.}
A  machine $M$ \emph{runs in (deterministic) time} $t\colon \N \rightarrow \N$,
if  $M$ reaches the output in $t(|x|)$
steps for each input $x\in \mathcal{I}$.
The  machine $M$ runs in \emph{polynomial time} if $t$ 
is a polynomial function.
 The complexity class
$\PTIME_\RE$ is defined as the set of all subsets of $\RE^*$ that
are decided by some  machine $M$ running in polynomial time.

\paragraph{Nondeterministic complexity classes.}
A language $L\sub \RE^*$ is \emph{decided nondeterministically}
 by a BSS machine $M$, if
 \[
 x \in L \quad\text{ if and only if }\quad f_M((x,x')) =1, \text{ for some $x'\in  \RE^*$}.
 \]
Here we assume a slightly
different input mapping $g_I:\mathcal{I}\to \mathcal{S}$, which places
an input $(x_1, \ldots ,x_n,x'_1, \ldots ,x'_m)$ in the state
\[(\ldots ,0,n,m,x_1, \ldots ,x_n,x'_1, \ldots ,x'_m,\ldots )\in \mathcal{S},\]
where the sizes of $x$ and $x'$ are respectively placed on the first two coordinates.
When we consider languages that a machine $M$ decides nondeterministically, we
call $M$ \emph{nondeterministic}. Sometimes when we wish to emphasize
that this is not the case, we call $M$ \emph{deterministic}.
Moreover, we say that $M$ is \emph{[0,1]-nondeterministic}, if
the guessed strings $x'$ are required to be from $[0,1]^*$.
L is \emph{decided in time} $t\colon \N \rightarrow \N$,
if, for every  $x \in L$, $M$ reaches the output $1$ in $t(|x|)$ steps for some $x'\in \RE^*$.
The machine \emph{runs in polynomial time} if $t$ is a polynomial function.
The class $\NP_\RE$ consists of those languages $L\subseteq \RE^*$
for which there exists a machine $M$ that nondeterministically decides $L$
in polynomial time.
Note that, in this case, the size of $x'$ above can be bounded by a polynomial (e.g., the running time of $M$) without altering the definition. The complexity class $\NP_\RE$ has many natural complete problems such as 4-FEAS, i.e., the problem of determining whether a polynomial of degree four has a real root \cite{blum1989}.

\paragraph{Complexity classes with Boolean restrictions.}
If we restrict attention to machines $M$ that may use only $c\in \{0,1\}$ in constant assignment nodes, then the corresponding complexity classes are denoted using an additional superscript $0$ (e.g., as in $\NP^0_\RE$).  Complexity classes over real computation can also be related to standard complexity classes. For a complexity class $\calC$ over the reals, the \emph{Boolean part} of $\calC$, written $\bp{\calC}$, is defined as $\{L\cap\{0,1\}^*\mid L\in \calC\}$.

\paragraph{Descriptive complexity.}
Similar to Turing machines, also BSS machines can be studied from the vantage point of descriptive complexity. To this end, finite $\RE$-structures are encoded as finite strings of reals using so-called rankings that stipulate an ordering on the finite domain. Let $\A$ be an $\RE$-structure over $\tau\cup\sigma$ where $\tau$ and $\sigma$ are relational and functional vocabularies, respectively. A \emph{ranking} of $\A$ is any bijection $\pi\colon\Dom(A)\to \{1, \ldots ,|A|\}$. A ranking $\pi$ and the lexicographic ordering on $\mathbb{N}^k$ induce a \emph{$k$-ranking} $\pi_k\colon\Dom(A)^k\to \{1, \ldots ,|A|^k\}$ for $k\in \mathbb{N}$.
Furthermore, $\pi$ induces the following encoding $\enc_\pi(\A)$. First we define $\enc_\pi(R^\A)$ and $\enc_\pi(f^\A)$ for $R\in \tau$ and $f\in \sigma$:
\begin{itemize}
\item Let $R\in \tau$ be a $k$-ary relation symbol. The encoding $\enc_\pi(R^\A)$ is a binary string of length $\lvert A \rvert^k$ such that the $j$th symbol in $\enc_\pi(R^\A)$ is $1$ if and only if $(a_1, \ldots ,a_k) \in R^\A$, where $\pi_k(a_1, \ldots ,a_k)=j$.
\item Let $f\in \sigma$ be a $k$-ary function symbol. The encoding $\enc_\pi(f^\A)$ is string of real numbers of length $\lvert A \rvert^k$  such that the $j$th symbol in $\enc_\pi(f^\A)$ is $f^\A(\vec{a})$, where $\pi_k(\vec{a})=j$.
\end{itemize}
The encoding $\enc_\pi(\A)$ is  then the concatenation of the string $(1,\ldots ,1)$ of length $|A|$ and the encodings of the interpretations of the relation and function symbols in $\tau\cup\sigma$.
We denote by $\enc(\A)$ any encoding $\enc_\pi(\A)$ of $\A$.

Let $\calC$ be a complexity class and $\eso{S}{O,E,C}$ a logic, where $O\sub \{+,\times, \SUM\}$,  $E \sub \{=,<,\leq\}$,  $C\sub \RE$, and $S\sub \RE$ or $S=d[0,1]$. Let $X\sub \RE$ or $X=d[0,1]$, and let $\calS$ be an arbitrary class of $X$-structures over $\tau\cup\sigma$ that is closed under isomorphisms.
We write $\enc(\calS)$ for the set of encodings of structures in $\calS$. Consider the following two conditions:
\begin{enumerate}[label=(\roman*)]
\item $\enc(\calS)=\{\enc(\A)\mid \A \in \struc^X(\phi)\}$ for some $\phi\in \eso{S}{O,E,C}[\tau\cup\sigma]\}$,
\item $\enc(\calS)\in \calC$.
\end{enumerate}
If  $(i)$ implies $(ii)$, we write $\eso{S}{O,E,C} \leq_X \calC$, and if the vice versa holds, we write $\calC \leq_X \eso{S}{O,E,C} $. If both directions hold, then we write $\eso{S}{O,E,C} \equiv_X \calC$. We omit the subscript $X$ in the notation if $X=\RE$.

The following results due to Gr\"adel and Meer extend Fagin's theorem to the context of real computation.
\begin{theorem}[\cite{GradelM95}]\label{thm:meer}
$\esor{+,\times,\leq,(r)_{r\in \RE}}\equiv \NP_\RE$. 
\end{theorem}

\end{document}